\documentclass[reprint, amssymb,superscriptaddress, aps,nofootinbib]{revtex4-1}

\usepackage{tikz}
\usepackage{hyperref}

\usepackage{graphicx}
\usepackage{dcolumn}
\usepackage{amssymb,amsmath}
\usepackage{epstopdf}
\usepackage[utf8]{inputenc}
\usepackage{xcolor}
\usepackage[english]{babel}
\usepackage{amsthm}
\usepackage{CJK}
\usepackage[utf8]{inputenc}
\usepackage[english]{babel}
\usepackage{algcompatible}
\usepackage{newfloat}
\usepackage{tabularx}
\usepackage[export]{adjustbox}
\usepackage{mathtools}
\usepackage{mdframed}
\DeclareFloatingEnvironment[
    fileext=loa,
    listname=List of Algorithms,
    name=ALGORITHM,
    placement=tbhp,
]{algorithm}

\newcommand{\p}{\mathcal{P}}

\newcommand{\R}{\mathbb{R}}

\newcommand{\pswap}{q}

\newtheorem{lemma}{Lemma}
\newtheorem{theorem}{Theorem}
\newcommand{\lmax}{l_{\mathrm{max}}}
\newcommand{\dmod}{D_{\mathrm{mod}}}
\newcommand{\g}{\tilde{g}}
\newcommand{\ra}{\tilde{r}}

\newtheorem{proposition}{Proposition}

\begin{document}
\author{Kaushik Chakraborty}
\email{k.chakraborty@tudelft.nl}
\affiliation{QuTech, 
Delft University of Technology,\\
Lorentzweg 1, 2628 CJ Delft, The Netherlands.}
\affiliation{Kavli Institute of Nanoscience, Delft University of Technology, \\ Lorentzweg 1, 2628 CJ Delft, The Netherlands}
\author{David Elkouss}
\affiliation{QuTech, 
Delft University of Technology,\\
Lorentzweg 1, 2628 CJ Delft, The Netherlands.}
\email{d.elkousscoronas@tudelft.nl}
\author{Bruno Rijsman}
\affiliation{QuTech, 
Delft University of Technology,\\
Lorentzweg 1, 2628 CJ Delft, The Netherlands.}
\email{brunorijsman@gmail.com}
\author{Stephanie Wehner}
\email{S.D.C.Wehner@tudelft.nl}
\affiliation{QuTech, 
Delft University of Technology,\\
Lorentzweg 1, 2628 CJ Delft, The Netherlands.}
\affiliation{Kavli Institute of Nanoscience, Delft University of Technology, \\ Lorentzweg 1, 2628 CJ Delft, The Netherlands}

\title{Entanglement Distribution in a Quantum Network, a Multi-Commodity Flow-Based Approach}

\begin{abstract}

We consider the problem of optimising the achievable EPR-pair distribution rate between multiple source-destination pairs in a quantum internet, where the repeaters may perform a probabilistic bell-state measurement and we may impose a minimum end-to-end fidelity as a requirement. We construct an efficient linear programming formulation that computes the maximum total achievable entanglement distribution rate, satisfying the end-to-end fidelity constraint in polynomial time (in the number of nodes in the network). We also propose an efficient algorithm that takes the output of the linear programming solver as an input and runs in polynomial time (in the number of nodes) to produce the set of paths to be used to achieve the entanglement distribution rate. Moreover, we point out a practical entanglement generation protocol which can achieve those rates.



\end{abstract}
\maketitle

\section{Introduction}

The quantum internet will provide a facility for communicating qubits between quantum information processing devices \cite{Van14, LSWK04, Kim08,WEH18}. It will enable us to implement interesting applications such as quantum key distribution \cite{bb14,E91}, clock synchronisation \cite{kkbj14}, secure multi-party computation \cite{CGS02}, and others \cite{WEH18}. To enable a full quantum internet the network needs to be able to produce entanglement between any two end nodes connected to the network \cite{MSLM13,Cal17,PKTT19,CRDW19}.

In this paper, we consider the problem of optimising the achievable rates for distributing EPR-pairs among multiple source-destination pairs in a network of quantum repeaters while keeping a lower bound on the end-to-end fidelity as a requirement. We propose a polynomial time algorithm for solving this problem and we show that, for a particular entanglement distribution protocol, our solution is tight and achieves the optimal rate. Our algorithm is inspired by \emph{multi-commodity flow} optimisation which is a very well-studied subject and has been used in many optimisation problems, including classical internet routing \cite{hu63}. In the context of a classical internet, a \emph{flow} is the total number of data packets, transmitted between a source and a destination per unit time (\emph{rate}). In this context, a \emph{commodity} is a \emph{demand}, which consists of a source, destination and potentially other requirements like the desired end-to-end packet transmission rate, quality of service, etc. In a classical network, a source and a destination can be connected via multiple communication channels as well as a sequence of repeaters and each of the communication channels has a certain \emph{capacity} which upper bounds the amount of flow it can transmit. In this context, a flow must satisfy another restriction, called \emph{flow conservation}, which says that the amount of flow entering a node (\emph{inflow}), except the source and destination node, equals the amount of flow leaving the node \footnote{Assuming that the intermediary repeater nodes do not lose packets while processing them.} (\emph{outflow}). With these constraints, one of the goals of a multi-commodity flow optimisation problem is to maximise the total amount of flows (end-to-end packet transmission rates) in a network given a set of commodities (demands). There exist \emph{linear programming} formulations for solving this problem and if we allow the flows to be a fraction then this \emph{linear programming} (LP) can be solved in polynomial  time (in the number of nodes) \cite{karm84}.




In a quantum internet, we abstract the entire network as a graph $G= (V,E, C)$, where $V$ represents the set of repeaters as well as the set of end nodes, and the set of edges, $E$, abstracts the physical communication links. Corresponding to each edge we define edge capacities $C: E\rightarrow \R^+$, which denotes the maximum elementary EPR-pair generation rate. We assume that the fidelity of all the EPR-pairs, generated between any two nodes $u,v \in V$ such that $(u,v) \in E$, is the same (say $F$). We refer to such EPR-pair as an \emph{elementary pair} and the physical communication link via which we create such an elementary pair is called an \emph{elementary link}. Flow in such a network is the EPR-pair generation rate between a source-destination pair. Depending on the applications, the end nodes may need to generate EPR-pairs with a certain fidelity. Keeping the analogy with the classical internet, here we refer to such requirement as a demand (commodity) and it consists of four items, a source $s \in V$, a destination $e \in V$, end-to-end desired entanglement distribution rate $r$ and an end-to-end fidelity requirement $F_{\text{end}}$. We denote the set of all such demands (commodities) as $D$. In this paper, we are interested in computing the maximum entanglement distribution rate (flow). Given a quantum network $G$ and a set of demands $D$, we investigate how to produce a set of paths $\p_i$ and an end-to-end entanglement generation rate $r_i$ (flow), corresponding to each demand $(s_i,e_i,F_i)$, such that the total entanglement generation rate $\sum_{i=1}^{|D|} r_i$ is maximised. In the rest of this paper, we refer to this maximisation problem as \emph{rate maximisation problem}. What is  more, here we also investigate what type of practical entanglement distribution protocol achieves such rate.

In the case of the quantum internet, we can use an LP for maximising the total flow $\sum_{i=1}^{|D|} r_i$. However, the working principle of quantum repeaters is different, unlike classical networks, the repeaters extend the length of the shared EPR-pairs by performing \emph{entanglement swapping} operations \footnote{Entanglement swapping is an important tool for establishing entanglement over long-distances. If two quantum repeaters, $A$ and $B$ are both connected to an intermediary quantum repeater $r$, but not directly connected themselves by a physical quantum communication channel such as fiber,  then $A$ and $B$ can nevertheless create entanglement between themselves with the help of $r$. First, $A$ and $B$ each individually create entanglement with $r$. This requires one qubit of quantum storage at $A$ and $B$ to hold their end of the entanglement, and two qubits of quantum storage at $r$. Repeater $r$ then performs an entanglement swap, destroying its own entanglement  with $A$ and $B$, but instead creating entanglement between $A$ and $B$. This process can be understood as repeater $r$ teleporting its qubit entangled with $A$ onto repeater $B$ using the entanglement that it shares with $B$.}  \cite{MATN15, teleport93, ZZHE93,GWZ08}. However, entanglement swapping operations might be probabilistic depending on the repeater technology used in the quantum internet. This implies that the usual flow-conservation property which we use in classical networks does not hold in the quantum networks, i.e., the sum of the inflow is not always equal to the sum of the outflow. Hence, the standard multi-commodity flow-based approach cannot be applied directly. 
For example, if the repeaters are built using \emph{atomic ensemble} and \emph{linear optics} then they use a probabilistic Bell-state measurement (BSM) for the entanglement swap operation \cite{SSDG11, GMLC13, SSMS14}. Due to the probabilistic nature of the BSM, the entanglement generation rate decays exponentially with the number of swap operations.

Another difficulty for using the standard multi-commodity flow-based approach for solving our problem occurs due to the end-to-end fidelity requirement in the demand. In a quantum network, the fidelity of an EPR-pair drops with each entanglement swap operation. This implies that a longer path-length results in a lower end-to-end fidelity. One can enhance the end-to-end fidelity using \emph{entanglement distillation}. However, some repeater technologies are unable to perform such quantum operations (for instance, the atomic ensemble-based quantum repeaters). Hence, for such cases, one can achieve the end-to-end fidelity requirement only by increasing the fidelity $F$ of the elementary pairs and reducing  the length of the discovered path. The first of these two options, the elementary pair fidelity, depends on the hardware parameters at fabrication. The second option is related to the path-length and it is under the control of the routing algorithm that determines the path from the source to the destination. For the routing algorithms, one possible way to guarantee the end-to-end fidelity is to put an upper bound on the discovered path-lengths. The standard multi-commodity flow-based LP-formulations does not take into account this path-length constraint. However, there exists one class of multi-commodity flow-based LP-formulation, called \emph{length-constrained multi-commodity flow} \cite{MM10}, which takes into account such constraints. In this paper, our proposed LP-formulation is inspired from the length-constrained multi-commodity flow problem and it takes into account the path-length constraint.


Given these differences, one might use the LP-formulation corresponding to the standard multi-commodity flow-based approach which we described before, but this would lead to a very loose upper bound on the achievable entanglement generation rate \cite{BAKE18,pir16,Pir19,Pir19bounds,AML16,AK17,RKBK18,BA17} in this setting. 

In our setting, it is not clear whether one can still have an efficient LP-formulation for the flow maximisation problem in a quantum internet. In fact, recently in \cite{li20} the authors mention that multi-commodity flow optimisation-based routing in a quantum internet may, in general, be an NP-hard problem. In this paper, we show that for some classes of practical entanglement generation protocols, one can still have efficient LP-formulation which maximises the total flow for all the commodities in polynomial  time (in the number of nodes).

The organisation of our paper is as follows: section \ref{sum_contr} provides a summary of our results. In section \ref{lp_form}, we give the exact LP-formulation for solving rate maximisation problem. In section  \ref{meth} we prove that our LP formulation solve the desired rate maximisation problem. Later, in the same section we show how one can achieve the entanglement generation rates proposed by the LP-formulation using an entanglement distribution protocol. We also analyse the complexity of our proposed algorithms in section \ref{meth}. We conclude our paper in section \ref{concl}.  

\section{Results}
\label{contr}

\subsection{Our contributions in a nutshell}
\label{sum_contr}
In this paper, all of our results are directed towards solving the rate maximisation problem in a quantum internet, where given a quantum network and a set of demands, the goal is to produce a set of paths such that the total end-to-end entanglement generation rate is maximised and in addition for each of the demands the end-to-end fidelity of the EPR-pairs satisfy a minimal requirement. 
In this section, we summarise our contributions.

\begin{itemize}
\item In order to solve the maximisation problem, we propose an LP-formulation called \emph{edge-based formulation} where both the number of variables and the number of constraints as well as the algorithm for solving such LPs scale polynomially with the number of nodes in the graph. However, it is non-trivial to see whether this formulation provides a valid solution to the problem or not. In this paper, by showing the equivalence between the edge-based formulation and another intuitive LP-formulation, called \emph{path-based formulation}, we show that the edge-based formulation provides a valid solution.


\item A disadvantage of the solution of the edge-based formulation is that it only gives the total achievable rate, not the set of paths which the underlying entanglement distribution protocol would use to distribute the EPR-pairs to achieve such rate. In this paper, we provide an algorithm, called the \emph{path extraction algorithm}, which takes the solutions of the edge-based formulation and for each of the commodities it extracts the set of paths to be used and the corresponding entanglement distribution rate along that path. The worst case time complexity of this algorithm is $O(|V|^4|E||D|)$, where $|V|, |E|$ denote the total number of nodes and edges in the network graph $G$ and $|D|$ denotes the total number of demands. What is more, we point out that there exists a practical entanglement distribution protocol along a path, called the \emph{prepare and swap protocol}, which achieves the rates (asymptotically) proposed by path extraction algorithm.  
\end{itemize}


\subsection{From the fidelity constraint to the path-length constraint}
\label{fid_to_len}

In a quantum network, the fidelity of the EPR-pairs drops with each entanglement swap operation. The fidelity of the output state after a successful swap operation depends on the fidelity of the two input states. If a mixed state $\rho$ has fidelity $F$, corresponding to an EPR-pair (say $|\Psi^+\rangle = \frac{1}{\sqrt{2}}(|00\rangle + |11\rangle)$) then the corresponding Werner state \cite{wern89} with parameter $W$ can be written as follows,

\begin{align*}
 &\rho = W|\Psi^+\rangle \langle ^+\Psi| + \frac{1-W}{4}~\mathbb{I}_4,
\end{align*}

where  $\mathbb{I}_4$ is the identity matrix of dimension $4$. The fidelity of this state is $\frac{1+3W}{4}$.

In this paper we assume that all the mixed entangled states in the network are Werner states. The main reason is that Werner states can be written as mixing with isotropic noise and hence form the worst case assumption. For the Werner states, if a node performs a noise-free entanglement swap operation between two EPR-pairs with fidelities $F$, then the fidelity of the resulting state is $\frac{1+3W^2}{4} $ which is equal to $1+ \frac{3}{4}\left(\frac{4F-1}{3}\right)^2$\cite{BDCZ98}.  

Here, each demand $(s_i,e_i,F_i)$ (where $1\leq i\leq |D|=k$) has $F_i$ as the end-to-end fidelity requirement. We assume that the fidelity of each of the elementary pairs is lower bounded by a constant $F$. Note that, in our model, we do not consider entanglement distillation, so in order to have a feasible solution, here we always assume that the fidelity requirement of the $i$-th demand, $F_i$ is at most the fidelity of the elementary pair $F > 0.5$. 
Corresponding to a demand $(s_i,e_i,F_i)$, if we start generating the EPR-pairs along a path $p = ((s_i,u_1), (u_1,u_2), \ldots , (u_{|p|-1},e_i))$, then the total number of required entanglement swap operations is $|p|-1$, where the path-length is $|p|$. As with each swap operation the fidelity drops exponentially, this implies the end-to-end fidelity will be $\frac{1+3W^{|p|}}{4} = 1+ \frac{3}{4}\left(\frac{4F-1}{3}\right)^{|p|}$. In order to satisfy the demand, $1+ \frac{3}{4}\left(\frac{4F-1}{3}\right)^{|p|}$ should be greater than $F_i$, i.e., $1+ \frac{3}{4}\left(\frac{4F-1}{3}\right)^{|p|} \geq F_i$. From this relation, we get the following constraint on the length of the path.

\begin{equation}
\label{path_length}
|p| \leq \left\lfloor\frac{\log \left(\frac{4F_i-1}{3}\right)}{\log \left(\frac{4F-1}{3}\right)}\right\rfloor. 
\end{equation}

This implies, for the $i$-th demand all the paths should have length at most $\left\lfloor\frac{\log \left(\frac{4F_i-1}{3}\right)}{\log \left(\frac{4F-1}{3}\right)}\right\rfloor$. In the rest of the paper, for the $i$-th demand we assume,

\begin{equation}
\label{len_const}
l_i := \left\lfloor\frac{\log \left(\frac{4F_i-1}{3}\right)}{\log \left(\frac{4F-1}{3}\right)}\right\rfloor.
\end{equation}

Using this constraint on the number of intermediate repeaters, we can rewrite the demand set $D$ in following way,

\begin{equation}
\label{eq:dem_len}
D = \{(s_1,e_1,l_1), \ldots , (s_k,e_k,l_k)\}.
\end{equation}


\subsection{LP-formulation}
\label{lp_form}

In this section, we construct the LP-formulation for computing the maximum flow in a quantum network. For the simplicity, we consider the network $G=(V,E,C)$ as a directed graph and construct all the LP-formulations accordingly. Note that, one can easily extend our result to an undirected graph, just by converting each of the edges which connects two nodes $u,v$ in the undirected graph into two directed edges $(u,v)$ and $(v,u)$.

For the entanglement distribution rate, here we let the achievable rate between two end nodes $s_i,e_i \in V$ along a repeater chain (or a path) $p = ((s_i,u_1), (u_1,u_2), \ldots , (u_{|p|-1},e_i))$ be $r_p$ such that,

\begin{equation}
\label{eq:rate}
r_p \leq (\pswap)^{|p|-1}\min\{C(s_i,u_1), \ldots ,C(u_{|p|-1},e_i)\}, 
\end{equation}

where $C(u,v)$ denotes the capacity of the edge $(u,v)\in E$ and $|p|$ is the length of the path and $q$ is the success probability of the BSM. Later, in section \ref{meth} we show that there exists a practical protocol called prepare and swap which achieves this rate requirement along a path. For an idea of such protocol we refer to the example of figure \ref{exam0}. In the next section, we give the LP-construction of the edge-based formulation.

\subsubsection{Edge-based formulation}
\label{lbf_net}

In this section, we give the edge-based formulation for solving the rate maximisation problem. In this formulation, we assign one variable to each of the edges of the network. As the total number of edges, $|E|$, in a graph of $|V|$ nodes scales quadratically with the number of nodes in the graph, the total number of variables is polynomial in $|V|$. This makes the edge-based formulation efficient. However, it is challenging to formulate the path-length constraint in this formulation. The main reason is that, an edge can be shared by multiple paths of different lengths and the variables of the edge-based formulation corresponding to that edge do not give any information about the length of the paths. In this paper, we borrow ideas from the length-constrained multi-commodity flow \cite{MM10} to handle this problem. In order to implement the length constraint we need to modify the network graph $G$ as well as the demand set $D$. In the next section, we show how to modify the network graph and the demand set. 
\vspace{0.2in}

\textbf{Network modification.} To implement the length constraint in the edge-based formulation, we define an expanded graph $G' = (V',E', C')$ from $G = (V,E,C)$ such that it contains $\lmax +1$ copies of each of the nodes, where $\lmax = \max\{l_1, \ldots , l_k\}$ and for all $1\leq i \leq k$, $l_i$ denotes the length constraint of the $i$-th demand $(s_i,e_i,l_i)$. For a node $u \in V$, we denote the copies of $u$ as $u^0,u^1, \ldots , u^{\lmax}$. We denote $V^j$ as the collection of the $j$-th copy of all the nodes. This implies, $V' = \bigcup_{j=0}^{\lmax}V^j$. For each edge $(u,v) \in E$ and for each $0\leq j < \lmax$, we connect $u^j \in V'$ and $v^{j+1} \in V'$ with an edge, i.e., $(u^j , v^{j+1}) \in E'$. For each edge $(u^j,v^{j+1}) \in E'$ we define $C'(u^j,v^{j+1}) := C(u,v)$. In figure \ref{exam3} we give an example of this extension procedure corresponding to a network graph of figure \ref{exam_glob}.

\begin{figure}
\centering
    \includegraphics[width=0.7\textwidth, left]{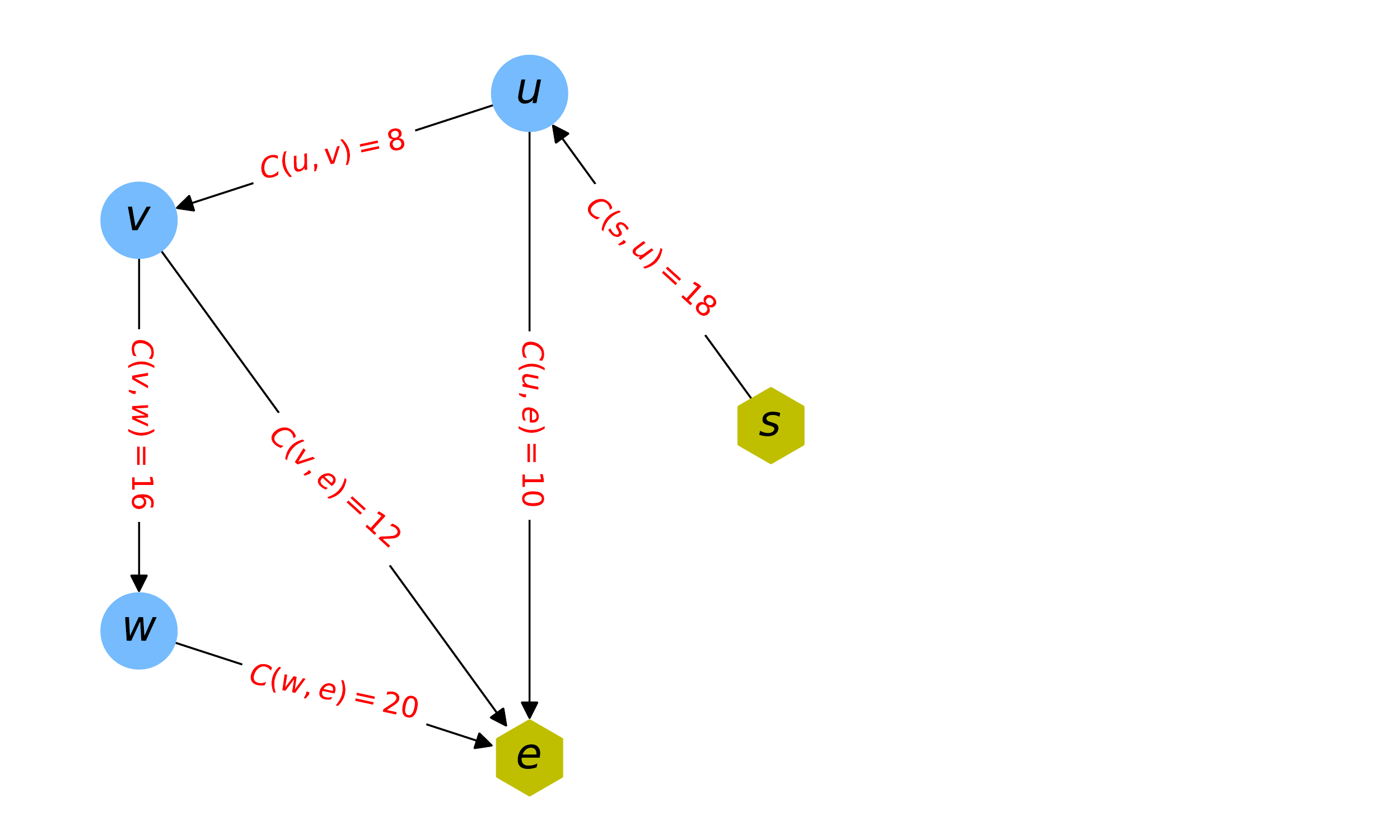}
    \caption{Network Graph $G = (V,E, C)$, with demand $D = \{(s,e,2)\}$ and $q = \frac{1}{2}$, i.e., here the source, $s$ wants to share EPR-pairs with $e$. In this network, for each edge $(u,v)\in E$, the quantity, $C(u,v)$ denotes the EPR-pair generation rate corresponding to that edge. }
    \label{exam_glob}
\end{figure}

\begin{figure}
    \includegraphics[width=0.8\textwidth, left]{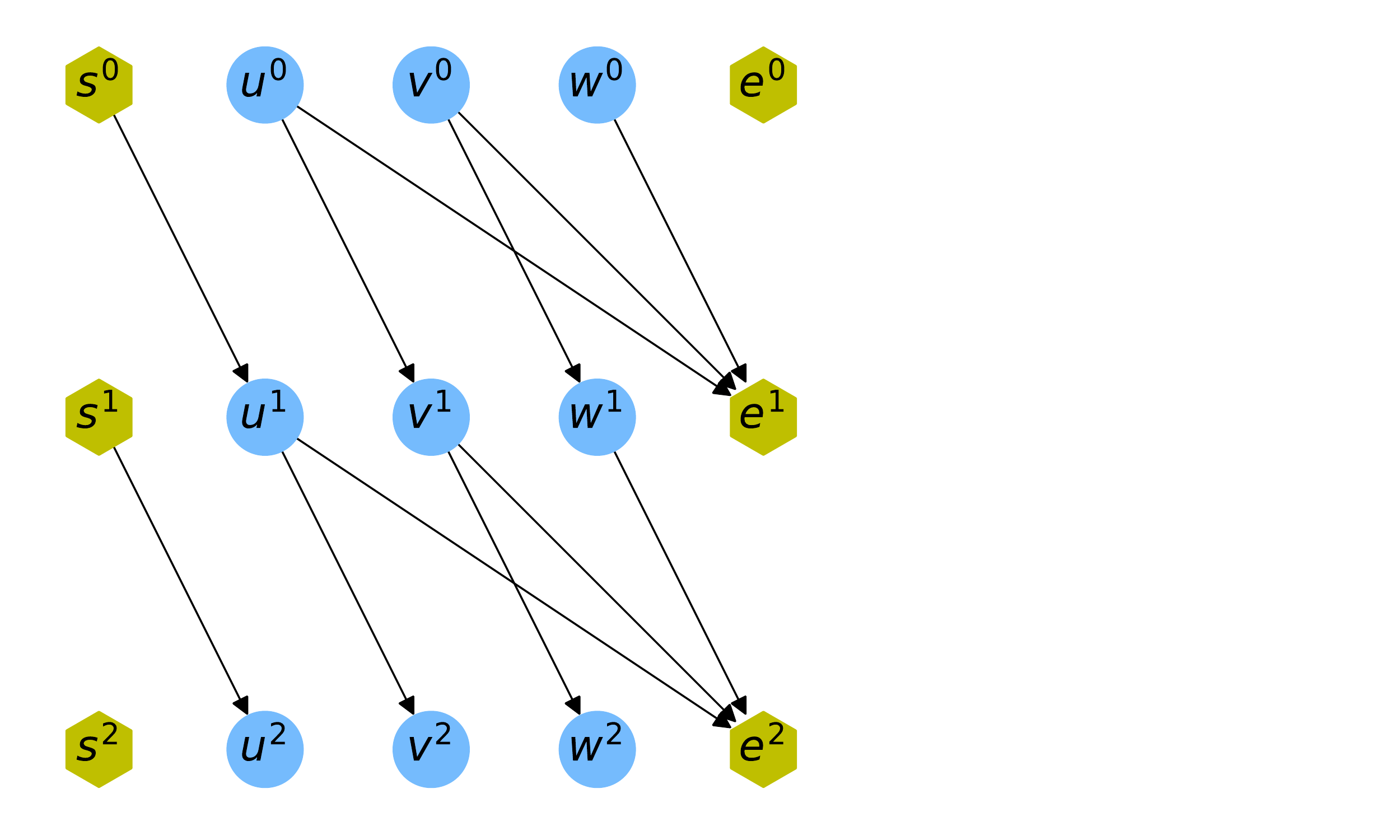}
    \caption{Extended Network Graph $G' = (V',E', C')$ of the original graph in figure \ref{exam_glob}. Here we are interested in finding the paths between $s$ and $e$ with path-length at most $2$. For the construction of $G'$, we create three copies $u^0,u^1,u^2$ of each of the nodes $u \in G$. There are five nodes $s,u,v,w,e$ in the graph of figure \ref{exam_glob}. This implies, in this modified graph we have $(s^0,s^1,s^2),(u^0,u^1,u^2), (v^0,v^1,v^2),(w^0,w^1,w^2),(e^0,e^1,e^2)$, fifteen nodes. In this original graph of figure \ref{exam_glob}, if $u$ is connected to $v$, then in this modified graph, we connect, $u^0$ with $v^1$ and $u^1$ with $v^2$. Note that, all the paths from $s^0$ to $e^1$ or $e^2$ has hop length at most $2$. In this modified graph the new demand set corresponding to the demand $D = \{(s,e)\}$ in the original graph $G$ in figure \ref{exam_glob} is $ D_{\mathrm{mod}} = \{(s^0,e^1),(s^0,e^2)\} $.}
    \label{exam3}
    \end{figure}

According to this construction, the length of all the paths in $G'$ from $s^0_i$ to $e^j_i$ is exactly $j$ and all the paths from $s^0_i$ to $e^1_i, \ldots , e^{l_i}_i$ have a path-length at most $l_i$. This implies, for the edge formulation, the $i$-th demand $(s_i,e_i,l_i)$ can be decomposed into $l_i$ demands $\{(s^0_i,e^1_i), \ldots , (s^0_i,e^{l_i}_i)\}$. This implies the total modified demand set would become,

\begin{equation}
\label{mod_dem}
\dmod := \{D_1, \ldots , D_{k}\},
\end{equation}

where for each $1 \leq i \leq k$, $D_i = \{(s^0_i,e^1_i), \ldots , (s^0_i,e^{l_i}_i)\}$. Note that, the new demand set $\dmod$ doesn't have any length or fidelity constraint. 


\vspace{0.2in}

\textbf{Edge-based formulation.}  Here, we give the exact LP-construction of the edge-based formulation. Note that, one can use a standard LP-solver (in this paper we use Python 3.7 pulp class \cite{LPPulp}) to solve this LP (see figure \ref{graph_exam_surf} for an example). 

In this LP-construction, for the $i,j$-th demand $(s^0_i,e^j_i) \in D_i$ (where $D_i \in \dmod$), we define one function $g_{ij}: E' \rightarrow \R^+$. The value of this function $g_{ij}(u,v)$, corresponding to an edge $(u,v) \in E'$ denotes the flow across that edge for the $i,j$-th demand. We give the edge-based formulation in table \ref{EBF}. 

\begin{table}[ht]
\begin{mdframed}
\begin{align}
\nonumber
&\text{Maximize}\\ \label{edge_form_plen3}
&\sum_{i=1}^k\sum_{j=1}^{l_i} (\pswap)^{j-1}\sum_{\substack{v^1:(s^0_i,v^1) \in E'}}  g_{ij}(s^0_i,v^1). \\ \nonumber
&\text{Subject to:}\\ \nonumber 
&\text{For all }  1 \leq i \leq k,  1 \leq j \leq  l_{i},  0 \leq t \leq \lmax-1, (u,v) \in E,\\  \label{const_edge31}
& g_{ij}(u^t,v^{t+1}) \geq 0, \text{ and} \\ \label{const_edge32}
& \sum_{i=1}^k \sum_{j=1}^{l_i} \sum_{t=0}^{\lmax -1}g_{ij}(u^t,v^{t+1}) \leq C(u,v).\\ \nonumber
&\text{For all } 1 \leq i \leq k, 1 \leq j \leq  l_{i}, u',v',w' \in V' : v \neq s^0_i, v \neq e^{j}_i , \\ \label{const_edge33}
&\sum_{u': (u',v') \in E'} g_{ij}(u',v')= \sum_{w': (v',w') \in E'} g_{ij}(v',w').
\end{align}
\end{mdframed}
\caption{Edge-based formulation}
\label{EBF}
\end{table}

In the objective function equation \ref{edge_form_plen3} of the edge-based formulation, the sum $\sum_{\substack{v^1:(s^0_i,v^1) \in E'}}g_{ij}(s^0_i,v^1)$ denotes the entanglement distribution rate between $s^0_i,e^j_i$, for a fixed $i,j$, when $(\pswap) = 1$. Note that, according to the construction of the graph $G'$, all the paths between $s^0_i,e^j_i$ have path-length exactly $j$. This implies that $(\pswap)^{j-1}\sum_{\substack{v^1:(s^0_i,v^1) \in E'}}g_{ij}(s^0_i,v^1)$ denotes the entanglement distribution rate between $s^0_i,e^j_i$ and $\sum_{i=1}^k\sum_{j=1}^{l_i} (\pswap)^{j-1}\sum_{\substack{v^1:(s^0_i,v^1) \in E'}} g_{ij}(s^0_i,v^1)$ denotes the total entanglement distribution rate for all the demands. The condition in equation \ref{const_edge32} represents the capacity constraint and condition equation \ref{const_edge33} denotes the flow conservation property. 

Although this construction is efficient, it does not give any intuition whether it will solve the rate maximisation problem. In section \ref{meth} we give another intuitive LP-formulation, called \emph{path-based formulation} and we explain the equivalence between the path-based and the edge-based formulation. This proof guarantees that the solution of the edge-based formulation gives a solution for the rate maximisation problem. 

\begin{figure}
    \includegraphics[width=0.7\textwidth, left]{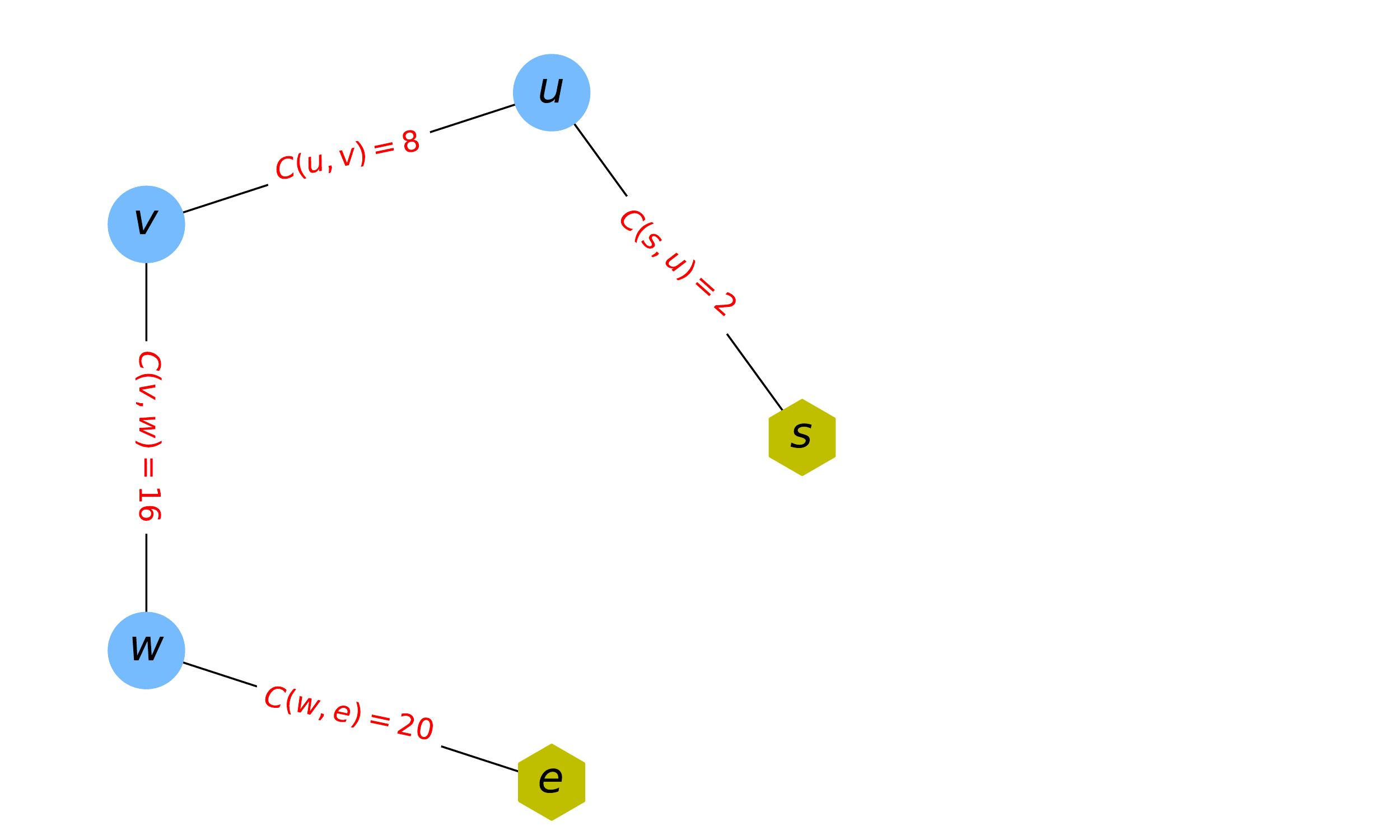}
    \caption{Repeater chain network with three intermediate nodes and one source-destination pair $s,e$. Here, there is a demand to create EPR-pairs between source $s$ and destination $e$. The capacity of an edge $(u,v)$ is given by the function $C(u,v)$. All the repeaters use BSM for the entanglement swap. Here we assume that the success probability of the BSM is $q = \frac{1}{2}$. 
    In the prepare and swap protocol, all the intermediate repeaters perform the swap operation at the same time. This implies, the expected entanglement generation rate between $s$ and $e$ for this protocol would be $r_{s,e} = (\pswap)^{4-1} \min\{C(s,u),C(u,v),C(v,w),C(w,e)\} =2(\frac{1}{2})^{3}  = 0.25$. }
    \label{exam0}
\end{figure}

\subsubsection{Path-extraction algorithm}
\label{sec:path_extrac}

The edge-based formulation, proposed in the last section is a compact LP construction and it can be solved in polynomial time. However, this solution only gives the total achievable rates for all the commodities, it does not give us any information about the set of paths corresponding to each commodity along which one should distribute the EPR-pairs to maximise the entanglement distribution rate. In this section, we give an efficient method for doing so in algorithm \ref{flow_dec3}, which takes the solution of the edge-based formulation and produces a set of paths as well as the achievable rates across each path for each of the demands. Later, in section \ref{meth} we show that the set of extracted paths satisfies the path-length constraint for each of the demands and if one uses the prepare and swap method for distributing entanglement across each of the paths then one can achieve the entanglement distribution rate suggested from this algorithm.

\begin{algorithm}
\begin{mdframed}
\vspace{0.1in}
\hspace*{\algorithmicindent} \textbf{Input: } The solution of the edge-based formulation, i.e., $\{\{g_{ij}(u',v')\}_{(u',v')\in E'}\}_{1\leq i\leq k, 1 \leq j \leq l_i}$. \\
    \hspace*{\algorithmicindent} \textbf{Output: } Set of paths as well as the rate across each of the paths $\{\p_{i,j}\}_{1\leq i\leq k, 1 \leq j \leq l_i}$.
\begin{algorithmic}[1]
 \FOR{($i=1$; $i \leq k$; $i++$)}
  \FOR{($j=1$; $j \leq l_i$; $j++$)}
  \State $m =0$.
  \State $F_{i,j}=g_{ij}$.
  \State $\p_{i,j} = \emptyset$.
 \WHILE{$\sum_{v:(s^0_i,v^1)\in E'} F_{i,j,m}(s^0_i,v^1) >0$}
  \State Find a path $p_{j,m}$ from $s^0_i$ to $e^j_i$ such that,
   \State $\forall(u',v')\in p_{j,m}$, $F_{i,j,m}(u',v') >0$
  \State $\ra_{p_{j,m}} = (\pswap)^{j-1}\{\min_{(u',v')\in p_{j,m}}\{F_{i,j,m}(u',v')\}$
   \State $\forall (u',v')\in p_{j,m}$,
  \State $F_{i,j,m+1}(u',v') = F_{i,j,m}(u',v') - \frac{\ra_{p_{j,m}}}{(\pswap)^{j-1}}$.
  \State $\p_{i,j} = \p_{i,j} \cup (p_{j,m},\ra_{p_{j,m}})$.
    \State $m=m+1$.
  \ENDWHILE
  \ENDFOR
\ENDFOR
\end{algorithmic}
\vspace{0.1in}
\end{mdframed}
 \caption{Path Extraction and Rate Allocation Algorithm.}
 \label{flow_dec3}
\end{algorithm}

\subsubsection{Example}

In this section, we give an example of our algorithms on a real world network topology $G = (V,E,C)$. In order to do so, we choose a SURFnet topology from the \emph{internet topology zoo} \cite{KNH11}. This is a publicly available example of a dutch classical telecommunication network, with $50$ nodes (see figure \ref{exam_surf}). In this network, we assume that each of the nodes in the network is an atomic ensemble and linear optics-based quantum repeater. These types of repeaters can generate elementary pairs almost deterministically \cite{GMLC13, SSMS14}, due to their multiplexing abilities. Here we assume that the elementary pair generation is a deterministic process. The elementary pair generation rate depends only on the entanglement source and its efficiency. Here, we choose the elementary pair generation rate uniformly randomly from $1$ to $400$ EPR-pairs per second. The success probability for the BSM, $q$ is considered to be $0.5$ for all the nodes. We also assume that the memory efficiency is one and all the memories are on-demand memories, i.e., they can retrieve the stored EPR-pairs whenever required \cite{GMLC13}. 

\begin{figure}[h]
    \includegraphics[width=0.48\textwidth, left]{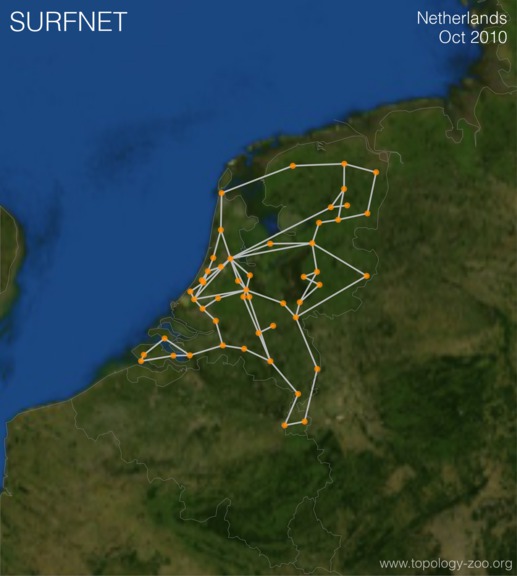}
    \caption{Pictorial view of the dutch SURFnet network, taken from internet topology zoo \cite{KNH11}. In this paper, we suppose that this is a quantum network and all the nodes in this network are quantum repeater nodes and some of them are the end nodes. We also assume that the repeaters can generate EPR-pairs using the communication links, shown in this figure. We run our proposed edge-based formulation and path-extraction algorithm on this network topology for maximising the total end-to-end EPR-pairs generation rate. We refer to figure \ref{graph_exam_surf} for detailed description.}
    \label{exam_surf}
\end{figure}


We additionally assume that the minimum storage time of all the memories is the maximum round trip communication time between any two nodes in the network that are directly connected by an optical fiber. In the SURFnet network, the maximum length of the optical fiber connecting any two nodes is $50$ km. Hence, the minimum storage time is $\frac{2\times 50000}{c} = 500 \mu s$, where $c$ is the speed of light in a telecommunication fiber, which is approximately $c \approx 2\times 10^8 $ meters per second. In this example, we consider the fidelity of all elementary pairs to be $F = 0.9925$. We generate the demands uniformly at random, i.e., we choose the sources and the destinations uniformly at random between $1$ and $50$. We also choose the end-to-end fidelity randomly from $0.93$ to $0.99$ for each of the demands. Substituting these fidelity constraints in equation \ref{len_const}, we obtain a maximum path-length $l_{\text{max}} = 8$. Here, we have generated only four demands and we assume that all the entanglement distribution tasks are performed in parallel. We optimise the total achievable rate using the LP solver available in the Python $3.7$ pulp class \cite{LPPulp}. The rates and the paths corresponding to the four demands are shown in figure \ref{graph_exam_surf}. An overview of the entire procedure is given in algorithm \ref{qpce} and the code of this implementation is publicly available in \cite{githubcode_rout20} and the data set can be found in \cite{data_set_routing_20}.

\begin{algorithm}
\begin{mdframed}
\vspace{0.1in}
\hspace*{\algorithmicindent} \textbf{Input: }Set of demands $D$, Network Graph $G = (V,E,C)$.\\
\hspace*{\algorithmicindent} \textbf{Output: }Set of paths $\p_i$ for the $i$-th demand and rate $r_p$, across each of the path $p\in \p_i$.
\begin{algorithmic}[1]
\State Convert the fidelity requirement $F_i$ of the $i$-th demand $(s_i,e_i,F_i)\in D$ into a path-length constraint $l_i$ (use equation \ref{len_const}).
 \State Compute the modified demand set $\dmod$ from $D$ according to the path-length constraint which we compute at the previous step (see subsection \ref{lbf_net} for details).
 \State Compute the extended network $G'$ from $G$ using the procedure, described in subsection \ref{lbf_net}.
 \State Implement the edge-based LP-formulation, proposed in table \ref{EBF} and compute the total maximum achievable rate $\sum_{i=1}^k r_i$ using the LP solver available in the Python $3.7$ pulp class \cite{LPPulp}.
 \State For the $i$-th demand, extract the set of paths $\p_i$ and compute the required rate $r_p$ across each of the paths $p\in \p_i$, such that $r_i = \sum_{p \in \p_i}r_p$ using the algorithm \ref{flow_dec3}.
 \end{algorithmic}
  \vspace{0.1in}
 \end{mdframed}
 \caption{Method to solve the rate maximisation problem.}
 \label{qpce}
\end{algorithm}

\begin{figure}
    \includegraphics[width=0.65\textwidth, left]{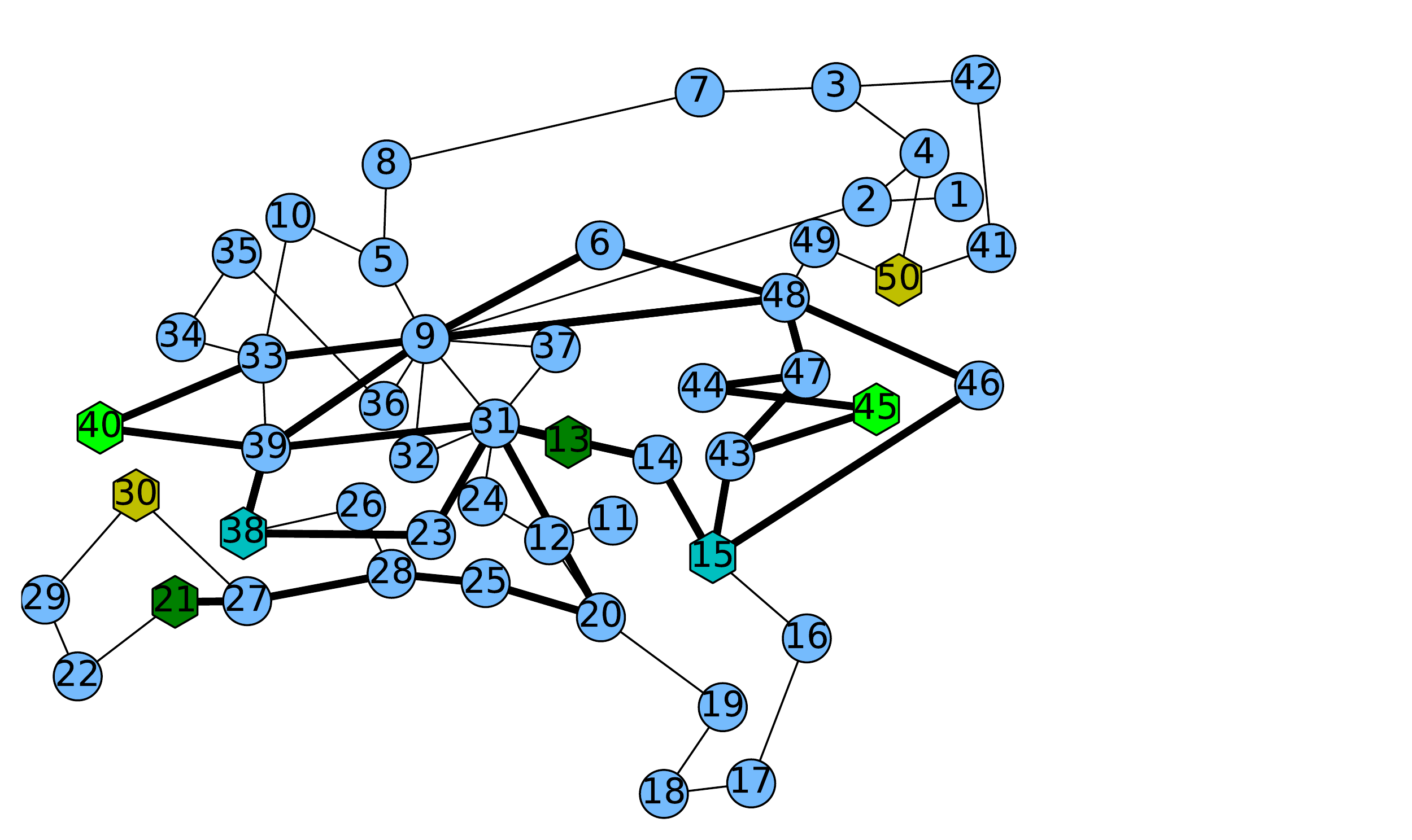}
    \caption{This figure is the graphical representation of the SURFnet network of figure \ref{exam_surf}. We run our rate maximisation algorithm on this graph. Here, we consider the demand set as $D = \{(40,45,7), (21,13,6),(30,50,7),(38,15,8)\}$. The end-nodes are represented using square boxes. The capacity of each of the link is chosen uniformly at random between $[1,400]$-EPR-pairs per second. We assume that the success probability of the bell-state measurement is $q = 0.5$ and the fidelity of the each of the elementary pairs is $F=0.9925$ and the value of the corresponding Werner parameter $W= 0.99$. Given such a network and demand set, we run our edge-based formulation proposed in table \ref{EBF} and get a total achievable rate of $18.140625$ EPR-pairs per second. Next, we feed the solution to the path extraction algorithm, proposed in algorithm \ref{flow_dec3} and extract the paths for each of the demands. In this figure, all the thick edges, participate in the paths. For the demand, $(45,40,7)$ we get two paths, $((45,43),(43,47),(47,48),(48,9),(9,33),(33,40))$ and $((45, 44), (44, 47), (47, 48), (48, 9), (9, 39), (39, 40))$. Note that each of the paths has path-length $6$. Hence the end-to-end fidelity for this demand is $\frac{1+3(0.99)^6}{4} = 0.955$. Similarly, for the demand $(21,13,6)$ the extracted path is $((21, 27), (27, 28), (28, 25), (25, 20), (20, 31), (31, 13))$ and end-to-end fidelity is $\frac{1+3(0.99)^6}{4} = 0.955$. For the demand $(30,50,7)$, there is no path of length smaller than $8$. Hence, this demand can not be satisfied in this model. For the demand $(38,15,8)$, we extract four paths, $((38, 23), (23, 31), (31, 14), (14, 15))$, and $((38, 39), (39, 31), (31, 14), (14, 15))$, and  $((38, 39), (39, 9),$ $(9, 48), (48, 46), (46, 15))$, $((38, 39), (39, 9), (9, 6), (6, 48),$ $(48, 47), (47, 43), (43, 15))$ and the end-to-end fidelity for the paths are $0.97$, $0.97$, $0.9625$ and $0.9475$. If we use the prepare and swap protocol for generating EPR-pairs, along each of the paths, then we can achieve the total entanglement distribution rate of $18.140625$ EPR-pairs per second.}
    \label{graph_exam_surf}
\end{figure}

\section{Methods}
\label{meth}

In this section, we provide more details of the results, presented in the last section. First we propose an intuitive LP-construction, called path-based formulation for solving the rate maximisation problem. Later, we give an idea about how to prove the equivalence between both of the proposed LP-constructions. Next, we explain the prepare and swap protocol in detail and show that using this protocol one can achieve the optimised entanglement generation rates along the paths from algorithm \ref{flow_dec3}. We finish this section with the complexity analysis of the edge-based formulation and algorithm \ref{flow_dec3}.

\subsection{Path-based formulation}
\label{p_form_prep_swap}

In this formulation, for each path $p$ corresponding to each source-destination pair $s_i,e_i$, we define one variable $r_p$. This $r_p$ denotes the achievable rate between $s_i,e_i$ along the path $p$. From equation \ref{eq:rate}, we have, for all $(u,v) \in p$,

\begin{equation}
r_p \leq (\pswap)^{|p|-1}C(u,v)  
\end{equation}

Note that, for the $i$-th demand the total rate $r_i$ can be achieved via multiple paths. Let $\p_i$ be the set of all possible paths which connect $s_i$ and $e_i$, hence,

\begin{equation}
r_i =\sum_{p \in \p_{i}} r_p.
\end{equation}

From equation \ref{eq:dem_len}, we have that for all the paths $p \in \p_i$, $|p|\leq l_i$. We give the exact LP-formulation which takes into account all these constraints in table \ref{path_form_normal}.

\begin{table}[ht]
\begin{mdframed}
\begin{alignat}{3}
\nonumber
&\text{Maximize}~~~~~~~~ \\ \label{simp_opt_cr_swap_len1}
&\sum_{i=1}^k  \sum_{p \in \p_i}r_p &\\ \nonumber
&\text{Subject to :}~~~~~~~\\ \label{eq:path_rate}
&\sum_{i=1}^k \sum_{\substack{p \in \p_i :\\ (u,v) \in p}}\frac{r_p}{(\pswap)^{|p|-1}} \leq C(u,v), &\forall ~(u,v) \in E, \\ \nonumber
&r_p \geq 0, &\forall i \in \{1, \ldots ,k\},\forall ~p\in \p_i, \\ \nonumber
&|p| \leq l_i &\forall i \in \{1, \ldots ,k\},\forall ~p\in \p_i.
\end{alignat}
\vspace{0.1in}
\end{mdframed}
\caption{Path-Based Formulation}
\label{path_form_normal}
\end{table}
Note that, in this LP-formulation, as we introduce one variable corresponding to each path so the total number of variables is of the order $O(|V|!)$. This scaling stops us from using the path-based formulation for solving the maximum flow problem. However, this formulation helps us to prove that the edge-based formulation gives a solution to the rate maximisation problem. In the next section, we give an idea of the proof. The full details of the proof are given in the supplementary material.

\subsection{Equivalence between the two formulations}
\label{ref:equiv}

The idea of the proof of equivalence is that, first we try to construct a solution of the edge-based formulation from the solution of the path-based formulation and then we try to construct a solution of the path-based formulation from the edge-based formulation. If both of the constructions are successful then we can conclude that both of the formulations are equivalent. 

In table \ref{path_form_normal} we provide the path-based formulation on the basis of the network graph $G = (V,E,C)$ and the demand set $D$, whereas in table \ref{EBF} we propose the edge-based formulation using the network graph $G' = (V',E',C')$ and the demand set $\dmod$. In order to show the equivalence between both of the formulations, we first need to rewrite the path-based formulation using the network graph $G' = (V',E',C')$ and the demand set $\dmod$. The next section focuses on this. After that, we focus on proving the equivalence.

\subsubsection{Path-based formulation on the modified network}
\label{p_form_prep_swap}

In this section, we construct the path-based formulation on the basis of the new demand set $\dmod$, defined in equation \ref{mod_dem} and the modified network $G' = (V',E',C')$. In this demand set, we use the term $i,j$-th demand to denote the $j$-th source destination pair $(s^0_i,e^j_i)$ of the $i$-th demand $(s_i,e_i) \in D$. We denote the set of all possible paths for the $i,j$-th demand as $\p_{i,j}$ and we assign a variable $r_p$, corresponding to each path $p \in \p_{i,j}$. From equation \ref{eq:rate} we have, for all the edges $(u',v') \in E'$ in a path $p$, 

\begin{equation}
r_p \leq (\pswap)^{|p|-1}C'(u',v').  
\end{equation}

Note that, from the construction of $G'$ and the new demand set $\dmod$, the length of all the paths in $\p_{i,j}$ is $j$. This implies, if $\p_i$ denotes the set of all possible paths for the $i$-th demand $(s_i,e_i,l_i)\in D$, then $\p_i = \bigcup_{j=1}^{l_i} \p_{i,j}$. For a fixed source-destination pair $(s_i,e_i)$, if along a path $p$, the achievable rate is $r_p$ then,

\begin{equation}
r_i = \sum_{j=1}^{l_i} \sum_{p \in \p_{i,j}} r_p.
\end{equation}

We give the exact formulation in table \ref{path_form}.

\begin{table}[ht]
\begin{mdframed}
\begin{alignat}{3}
\nonumber
&\text{Maximize} \quad \quad \quad \quad \quad \quad \quad \quad\\ \label{opt_cr_swap_len1}
&\sum_{i=1}^k \sum_{j=1}^{l_i}\sum_{p \in \p_{i,j}}r_p \quad \quad \quad \quad \quad \quad \\ \nonumber
&\text{Subject to :}\quad \quad \quad \quad \quad \quad \quad \quad \\ \label{swp_const1}
&\sum_{i=1}^k \sum_{j=1}^{l_i} \sum_{t=0}^{\lmax -1}\hspace{-0.13in}\sum_{\substack{p \in \p_{i,j} :\\ (u^t,v^{t+1}) \in p}}\hspace{-0.13in}\frac{r_p}{(\pswap)^{|p|-1}} \leq C(u,v), ~ &\forall ~(u,v) \in E,\\\ \nonumber
&\forall i \in \{1, \ldots ,k\}, \forall  j \in \{1, \ldots, l_i\}, \forall ~p\in \p_{i,j}, &\\
&r_p \geq 0.\quad \quad \quad \quad \quad \quad \quad \quad&
\end{alignat}
\vspace{0.1in}
\end{mdframed}
\caption{Path-Based Formulation on the Modified Network}
\label{path_form}
\end{table}

\subsubsection{Path-based formulation to edge-based formulation}
\label{pe_prep_swap}

In this section we construct a solution of the edge-based formulation from the solution of the path-based formulation. 
In the edge-based formulation, we construct a new demand set $\dmod$, where the $i$-th demand $(s_i,e_i,l_i)$ in the original demand set $D$ is decomposed into $l_i$-demands $(s^0_i,e^1_i), \ldots , (s^0_i,e^{l_i}_i)$. Recall that, the quantity $l_i$ denotes the upper bound on the discovered path-length which reflects the lower bound on the required end-to-end fidelity of the EPR-pairs generated between $s_i$ and $e_i$ (See section \ref{fid_to_len} for the details).  Here, each of the $(s^0_i,e^j_i)$ are the nodes in the modified graph $G'$. From the construction of $G'$, it is clear that all the paths from $s^0_i$ to $e^j_i$ have length $j$. If we have the solutions of the path formulation, proposed in table \ref{path_form}, then from there, for each edge $(u',v') \in E'$ and for the $i,j$-th demand $(s^0_i,e^j_i)$ we define the value of $\g_{ij}(u,v)$ as,

\begin{equation}
\label{eq_gij_prep_swap}
\g_{ij}(u',v') := \sum_{\substack{p \in P_{i,j}, \\ (u',v') \in p}} \frac{r_p}{(\pswap)^{j-1}}. 
\end{equation}

One can easily check that the definition of $\g_{ij}$, defined in equation \ref{eq_gij_prep_swap} satisfies all the constraints of the edge-based formulation, proposed in the equations \ref{const_edge31}, \ref{const_edge32}, \ref{const_edge33}. Moreover, with this definition of the $\g_{ij}$, the objective function (equation \ref{edge_form_plen3}) of the edge-based formulation becomes same as the objective function of the path-based formulation. This shows that, the optimal value of the edge-based formulation is at least as good as the solution of the path-based formulation.

\subsubsection{Edge-based formulation to path-based formulation}
\label{etp_prep_swap}

Here, we construct a solution of the path-based formulation from the solution of the edge-based formulation. We use the algorithm, proposed in algorithm \ref{flow_dec3} for extracting the paths and corresponding rate along that path. In the algorithm we compute the rate $\ra_{p_{j,m}}$ corresponding to a path $p_{j,m}$ for a demand $(s^0_i,e^j_i) \in \dmod$ as follows,

\begin{equation}
\label{prep_swap_rate}
\ra_{p_{j,m}} := (\pswap)^{j-1}\min_{(u',v')\in p_{j,m}}\{F_{i,j,m}(u',v')\},
\end{equation} 

where the function $F_{i,j,m}(u',v')$ is related to $g_{i,j}(u',v')$ and $\ra_{p_{j,m}}$. Here, $F_{i,j,0} = g_{i,j}$ and for all $m \geq 0$ we compute $F_{i,j,m+1}$ as follows,

\begin{equation}
\label{prep_swap_update}
\forall (u,v) \in p_{j,m}~~F_{i,j,m+1}(u',v')= F_{i,j,m} (u',v')- \frac{\ra_{p_{j,m}}}{(\pswap)^{j-1}}. 
\end{equation}

We give a detailed proof of the fact that the paths as well as the allocated rate corresponding to each path, extracted from algorithm \ref{flow_dec3} corresponds to the feasible solution of the path-based formulation in the appendix \ref{app_etp_prep_swap}. Moreover, if we consider the equation \ref{prep_swap_rate} as the definition $\ra_{p_{j,m}}$ then the objective function of the edge-based formulation is same as the objective function of the edge-based formulation. This shows that this is a valid solution of the path-based formulation. In the last section we showed that the solution of the path-based formulation is at least as good as the solution of the edge-based formulation. Hence, the solutions of both of the formulations are equivalent. 

\subsection{Prepare and swap protocol and the LP-formulations}

In this section, we explain the prepare and swap protocol and show that with this protocol one can achieve the entanglement distribution rate along a path, proposed by algorithm \ref{flow_dec3}. In the next section we explain the protocol for a repeater chain with a single demand. After that we extend the protocol for the case for multiple demands.

\subsubsection{Prepare and swap protocol for a repeater chain}
Suppose in a repeater chain $u_0=s, u_1, \ldots, u_n, u_{n+1} = e$, where for all $0 \leq i \leq n$, the nodes $u_{i},u_{i+1}$ are neighbours of each other and $s$ would like to share EPR-pairs with $e$. In this protocol, first, all the repeaters generate entanglement with its neighbours/neighbour in parallel and store the entangled links in the memory. Here we assume that entanglement generation across an elementary link is a deterministic event, i.e., the entanglement generation probability per each attempt is one. An intermediate node $u_i$ which resides between $u_{i-1}$ and $u_i$, performs the swap operation when both of the EPR-pairs between $u_{i-1},u_i$ and $u_i, u_{i+1}$ are ready. As we assume that each of the swap operations is probabilistic, so the entanglement generation rate with this protocol is lower. However, due to the independent swap operations, the protocol doesn't need a long storage time. This makes the protocol more practical.

We give an example of such an entanglement generation protocol on a repeater chain with three intermediate nodes and one source-destination pair in figure \ref{exam0}. In the next lemma, we derive an analytical expression of the end-to-end entanglement generation rate in a repeater chain network for the prepare and swap protocol. Note that, a variant of the proof of lemma \ref{egr_prep_swap} can be found in the literature \cite{SSMS14}. For completeness, in the supplementary material we include the proof of this lemma. 

\begin{lemma}
\label{egr_prep_swap}
In a repeater chain network with $n+1$ repeaters $\{u_0, u_1, \ldots, u_{n}\}$, if the probability of generating an elementary pair per attempt is one, the probability of a successful BSM is $(\pswap)$, the capacity of an elementary link $(u_i,u_{i+1})$ (for $0 \leq i \leq n-1$) is denoted by $C_i$ and if the repeaters follow the prepare and swap protocol for generating EPR-pairs, then the expected end-to-end entanglement generation rate $r_{u_0,u_{n}}$ is,
\begin{align}
\label{eg_rate_prep_mes}
r_{u_0,u_{n}} & = (\pswap)^{n-1}\min\{C_0, \ldots ,C_{n-1}\}.
\end{align}
\end{lemma}

Notice that the EPR-pair generation rate for this protocol is exactly same as the EPR-pair generation rate, proposed in equation \ref{eq:rate}, which we use for the path-based formulation. In the previous sections we show that both of the path-based and the edge-based formulations are equivalent. This implies, the rates extracted from algorithm \ref{flow_dec3} can be achieved with this prepare and swap protocol. 

\subsubsection{Prepare and swap protocol for an arbitrary network}

In a quantum network if there are multiple demands then one link might be shared between multiple paths. From algorithm \ref{flow_dec3} we get the set of paths and the desired EPR-pair generation rate $r_p$, across each of the paths $p$, passing through an edge $(u,v)$. In this scenario we use the prepare and swap protocol for each of the paths in a sequential manner or round-robin manner\footnote{One can use a more sophisticated scheduling algorithm for allocating the EPR-pair generation resources across an elementary link. The details of this scheduling algorithm are beyond the scope of this paper.}. From lemma \ref{egr_prep_swap} we get that, to achieve the rate $r_p$, we need to generate on average $\frac{r_p}{(\pswap)^{|p|-1}}$ elementary EPR-pairs per second across each of the elementary links $(u,v)$  along the path $p$. Here we also assume that the elementary link generation is a deterministic event, i.e., the two nodes $u$ and $v$ connecting the elementary link $(u,v)$ can generate exactly $C(u,v)$ EPR-pairs per second. Hence, each of the paths uses the elementary link $(u,v)$ for $\frac{r_p}{(\pswap)^{|p|-1}C(u,v)}$ seconds (on average) for generating the required elementary EPR-pairs\footnote{Note that, if $\frac{r_p}{(\pswap)^{|p|-1}}$ is a rational number then one can always achieve this average rate. For an irrational value of $\frac{r_p}{(\pswap)^{|p|-1}}$ we need to approximate it to a rational number.}. For generating $r_p$ EPR-pairs per second, all the nodes along the path $p$ need to use the elementary links for at most $\max_{(u,v)\in p}\left\{\frac{r_p}{(\pswap)^{|p|-1}C(u,v)}\right\}$ seconds, which is equal to $\frac{r_p}{(\pswap)^{|p|-1}}\frac{1}{\min_{(u,v)\in p}\{C(u,v)\}}$ seconds. This gives an upper bound on the average storage time of the quantum memory for generating the EPR-pairs along the path $p$. In figure \ref{exam1} we give an example of allocating EPR-pair generation resources for multiple demands.


\begin{figure}
    \includegraphics[width=0.7\textwidth, left]{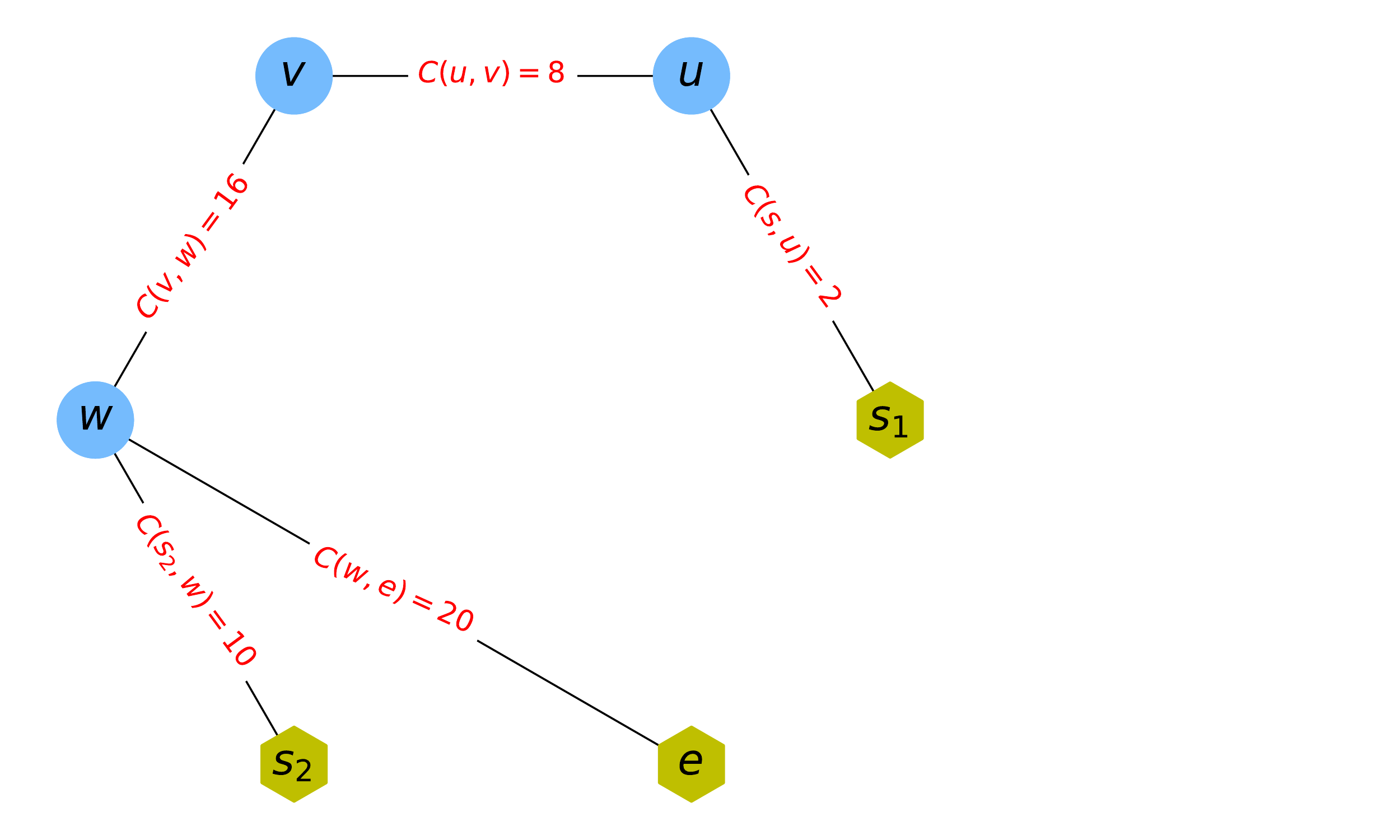}
    \caption{Repeater network $G = (V,E,C)$ with the demand set $D = \{(s_1,e),(s_2,e)\}$. All the repeaters use BSM for the entanglement swap. Here we assume that the success probability of the BSM is $q = \frac{1}{2}$. From the algorithm \ref{flow_dec3} we get a path $p_1 = ((s_1,u),(u,v),(v,w),(w,e))$ for the demand $(s_1,e)$ and a path $p_2 = ((s_2,w),(w,e))$ for the demand $(s_2,e)$. The desired rate across $p_1$ is $r_{p_1} = 0.25$ EPR-pairs per second and the desired rate across $p_2$ is $r_{p_2} = 5$ EPR-pairs per second. 
If we use the prepare and swap protocol along both of the paths separately then we get the desired EPR-pairs generation rate. However, here the elementary link $(w,e)$ is being shared by both of the paths. In this case, a simple scheduling technique would be to distribute the EPR-pairs sequentially. For example, at the beginning the elementary link $(w,e)$ generates the EPR-pairs for the demand $(s_1,e)$. For generating on average $0.25$ EPR-pairs across the path $p_1$, the elementary link $(w,e)$ has to generate $0.25 \times 2^{|p_1|-1} = 2$ EPR-pairs per second. As the capacity of the elementary link is $20$ EPR-pairs per second. Hence, it can generate $2$ EPR-pairs within $\frac{2}{20} = 0.1$ seconds (on average). Then it can start generating the EPR-pairs for the demand $(s_2,e)$. For this demand, the elementary link $(w,e)$ has to generate $5 \times 2^{|p_2|-1} = 10$ EPR-pairs per second. Hence, it will take on average $\frac{10}{20} = 0.5$ seconds to generate $10$ EPR-pairs.}
    \label{exam1}
\end{figure}

\subsection{Complexity analysis}
\label{complex_seq_swap_wo_fid}

In this section, we analyse the complexity of the LP formulations as well as the path-extraction and rate allocation algorithm (algorithm \ref{flow_dec3}). The edge-based LP-formulation, proposed in this paper is based on the modified network graph $G'$ and modified demand set $\dmod$ and the running time of the edge-based LP-formulation solver depends on the size of this network and modified demand set. In the next lemma we give an upper bound on the size of $G'$ and $\dmod$.

\begin{lemma}
\label{lem:graph_size}
The edge-based formulation, proposed in equations \ref{edge_form_plen3} to \ref{const_edge33}, has at most $N = |D||E||V|$ variables and $M = |V|^2|E||D| + |V||E| + |V|^2|D|$ constraints, where $|V|, |E|, |D|$ denote the total number of repeater nodes, total number of edges in the network graph $G$ and size of the demand set $D$ respectively.
\end{lemma}

\begin{proof}
The edge-based formulation, is based on the modified network $G' = (V',E',C')$ which we construct from the actual internet network $G = (V,E,C)$. In $G'$ we create at most $l_{\text{max}}$ copies of each of the nodes and edges. This implies, $|V'| \leq \lmax |V|$ and $|E'| \leq \lmax|E|$. As, $\lmax = \max \{l_1, \ldots, l_{|D|}\}$, hence $l_{\max} \leq |V|$. This implies, $|V'| \leq |V|^2$ and $|E'| \leq |E||V|$. In the construction of the edge-based formulation, for each demand, we introduce one variable corresponding to each edge of $E'$. Hence, the total number of variables for this formulation is $N = |D||E'| = |D||E||V|$.

In the edge-based formulation, the constraint equation \ref{const_edge31} ($g_{ij}(u,v) \geq 0$) holds for all $1\leq i \leq |D|$, for all $1 \leq j \leq l_i$ and for all edge $(u,v) \in E'$. This implies, the total number of constraints corresponding to equation \ref{const_edge31} is $|V|^2|E||D|$. Similarly, the constraint equation \ref{const_edge32} holds for all the edges $|E'|$. This implies, there are at most $|V||E|$ constraints corresponding to equation \ref{const_edge32}. The constraint equation \ref{const_edge33} should be satisfied by all the nodes in $V'$ and for all $1\leq i \leq |D|$, for all $1 \leq j \leq l_i$. This implies, the total number of constraints corresponding to that equation is $|V|^3|D|$. Hence, the total number of constraints in the edge-based formulation is $M = |V|^2|E||D| + |V||E| + |V|^3|D|$.

\end{proof}

The previous lemma implies that the total number of variables and constraints in the LP-formulation corresponding to the maximum flow problem scales polynomially with the number of nodes in the network graph $G$. This implies, the time complexity of the LP solvers for this problem also scales polynomially with the number of nodes in the network graph $G = (V,E,C)$. Hence, one can compute the maximum achievable rate for the quantum internet in polynomial time. Now we focus on the complexity of algorithm \ref{flow_dec3}. The algorithm \ref{flow_dec3} uses the solution of the edge-based formulation for extracting the set of paths for each of the demands. In the next proposition, we show that the size of the set of extracted paths corresponding to each demand is upper bounded by $|V||E|$. Then we use this result for computing the running time of algorithm \ref{flow_dec3}. 

\begin{proposition}
\label{prop_prep_cond2}
In algorithm \ref{flow_dec3},
\begin{equation}
|\p_{i,j}| \leq |E||V|.
\end{equation}
\end{proposition}

\begin{proof}
Due to the flow conservation property of the edge-based formulation, if for some neighbour of $s^0_i$, $F_{i,j,m}(s^0_i,v^1) > 0$ then there exist a path $p_{j,m}$ from $s^0_i$ to $e^j_i$ such that $F_{i,j,m}(u',v') >0$ for all $(u',v') \in p_{j,m}$. Note that, at each step $m$ of the algorithm \ref{flow_dec3} there exist at least one edge $(u',v') \in E'$ in the discovered path $p_{j,m}$, such that $F_{i,j,m+1}(u',v') =0$. As there are in total, $|E'|$ number of edges and the algorithm runs until $\sum_{v^1:(s^0_i,v^1) \in E'}F_{i,j,m+1}(s^0_i,v^1)=0$, so the maximum value of $m$ could not be larger than $|E'|$. From the construction of the modified network, $G'$ we have $|E'| \leq |V||E|$. This implies, $|\p_{i,j}| \leq |E'|\leq |E||V|$.
\end{proof}

In the next theorem we show that, running time of algorithm \ref{flow_dec3} is $O(|D||V|^4|E|)$.
  
\begin{theorem}
\label{thm_comp_rate_fid_prep}
The algorithm \ref{flow_dec3} takes the solution of the edge-based LP-formulation and extract the set of paths in $O(|D||V|^4|E|)$ time, where $|D|$ is the size of the demand set, $|V|, |E|$ denote the total number of nodes and edges in the network $G = (V,E,C)$.
\end{theorem}
\begin{proof}
In algorithm \ref{flow_dec3} we compute the paths based on the modified network $G' = (V',E',C')$, which we construct from the original network $G = (V,E,C)$. In this modified network $|V'| \leq |V|l_{\max} \leq |V|^2$.
In algorithm \ref{flow_dec3} at step $11$ we compute a path in the graph $G'$. Note that, in the worst case, it takes $O(|V'|) = O(|V|^2)$ time to find a path between a source-destination pair in a network. According to proposition \ref{prop_prep_cond2}, we have that for a fixed $1\leq i \leq |D|$ and a fixed $1 \leq j \leq l_i$ the total number of paths discovered by algorithm \ref{flow_dec3} is upper bounded by $O(|V||E|)$. Hence, for that $i,j$, the running time of algorithm \ref{flow_dec3} is $O(|V|^3|E|)$. As, $i \leq |D|$ and $j \leq l_i \leq l_{\max} \leq |V|$, so in the worst case scenario, the total running time of algorithm \ref{flow_dec3} is upper bounded by $O(|D||V|^4|E|)$. This concludes the proof.
\end{proof}

\section{Conclusion}
\label{concl}
In this paper, we use techniques from the length constrained multi-commodity flow theory for developing a polynomial-time algorithm for maximising achievable expected entanglement generation rate between multiple source-destination pairs in a quantum internet. Here, we have maximised the end-to-end entanglement distribution rate, satisfying a constraint on the end-to-end fidelity for each of the demand. We have shown that our LP-formulation provides a maximal solution if it exists. Our path extraction algorithm produces a set of paths and the achievable rates along each of the paths. The path-extraction algorithm has high running time as a function of the path-length. Here we consider the worst case scenarios, where we assume that the length of the discovered path scales with $|V|$. In practical scenarios, without distillation, the end-to-end fidelity of the distributed EPR-pairs would drop drastically with the path-length. Hence, it is fair to consider that the length of the allowed path increases slowly with the size of the network node set. This would make the path-extraction algorithm faster. 

One can use any entanglement generation protocol for distributing EPR-pairs across the paths that are discovered by the path-extraction algorithm. However, our LP-formulation is inspired from the atomic ensemble and linear optics based quantum repeaters, where the storage time is very short and the entanglement swap operation is probabilistic in nature \cite{SSDG11, GMLC13, SSMS14}. Here, we have also pointed out that, there exists a practical protocol, called prepare and swap protocol, which can be implemented using atomic ensemble based repeaters and if one uses this protocol for distributing entanglement across each of the paths, then one can generate EPR-pairs with the rate proposed by our path-extraction algorithm.

In this paper, we focus on maximising the end-to-end entanglement generation rate. However, one can easily extend our results for other objective functions, like minimising the weighted sum of congestion at edges. 

In future work, it would be interesting to include the more realistic parameters like the bounded storage capacity, time to perform the swap operation, etc., in our model and modify our current formulations to come up with more sophisticated routing algorithms. 

The proposed LP-formulations give an optimal achievable EPR-pairs distribution rate with respect to prepare and swap protocol. This protocol is practical and require very less amount of quantum storage time. However, there exist more sophisticated protocol which can achieve higher EPR-pairs distribution rate but require higher quantum storage time. Another interesting future research direction would be to find out a protocol for distributing EPR-pairs along a chain which achieves the optimal EPR-pair generation rate and find an LP-formulation for such protocol.

\begin{acknowledgements}
We would like to acknowledge W. Kozlowski for many stimulating discussions. We would like to thank M. Skrzypczyk for giving useful feedback on the draft. This publication is supported by an ERC Starting grant and the QIA-project that has received funding from the European Union's Horizon 2020 research and innovation program under grant Agreement No. 820445.
\end{acknowledgements}
\bibliographystyle{apsrev4-1}

\bibliography{mybibnet,revtex-custom}

\appendix
\section{Outline}
In the first part of the appendix, we focus on giving a detailed proof of the equivalence of the path-based and the edge-based formulation. In the second part of the appendix, we show how the prepare and swap protocol can achieve the entanglement distribution rate, which we get as an output from the LP solver. Before going to the detailed proof, in the next appendix first, we define again some of the notations which we use in the proof. One can find the equivalence of the edge-based formulation and the path-based formulation in appendix \ref{app_equiv_prep}. More precisely, for the clarity, in appendices \ref{app_path_form} and \ref{app_edge_form} we rewrite the path-based formulation and the edge-based formulations. We prove the equivalence between both of the formulations by showing that one can construct the solution of the edge-based formulation from the path-based formulation (see appendix \ref{app_pte_prep_swap}) and vice-versa (see appendix \ref{app_etp_prep_swap}). In appendix \ref{app_prep_swap}, we describe the entanglement distribution rate for the prepare and swap protocol across a path.


\section{Notations}
\label{not}

In this section, we define again some of the notations we are going to use later in the proofs. We first start with the network graph $G = (V, E, C)$, which is a directed graph, and it abstracts the quantum network. Here $V$ denotes the set of quantum repeaters, $E$ denotes the set of quantum communication links and $C: E \rightarrow \R^+$ denotes the entanglement generation capacity of an edge $(u,v) \in E$, i.e., how many EPR-pairs the nodes $u,v$ can generate per second. A path $p$ between a source node and a destination node in the graph is a finite sequence of edges which joins a sequence of distinct vertices. The path-length of a path $p$ between a source-destination pair $(s,e)$ is denoted by $|p|$. 
Next, we define the set of demands $D = \{(s_1,e_1,l_1), \ldots , (s_k,e_k,l_k)\}$, where the $i$-th element of this set (or $i$-th demand) is a triplet $(s_i,e_i,l_i)$ and $s_i$ would like to share EPR-pairs with $e_i$ using the multiple paths, whose path-lengths are at most $l_i$. Here, we assume the size of the demand set $|D| =k$.


\section{Equivalence of the path-based and edge-based formulation for the prepare and swap protocol}
\label{app_equiv_prep}

In this section, we prove the equivalence between the path-based formulation and the edge-based formulation. Before the proof, for clarity, in the next two subsections, we rewrite the path-based formulation and the edge-based formulation. Note that for the edge-based formulation of the LP construction, we use the modified network graph $G'= (V',E',C')$, which is constructed from $G = (V,E,C)$. For the clarity here again we rewrite the construction of $G'$. First, from the demand set, $D$ we compute $\lmax = \max\{l_1, \ldots , l_k\}$, where each of the $l_i$ is related to the length constraint of the $i$-th demand $(s_i,e_i,l_i)$. Then, for each node $u \in V$, we create $\lmax + 1$ copies of $u$. We denote them as $u^0,u^1, \ldots , u^{\lmax}$. We denote $V^j$ as the set of the $j$-th copy of all the nodes in $V$, i.e., $V^j := \{u^j: u\in V\}$. For $G'$, the set of nodes $V' = \bigcup_{j=0}^{\lmax}V^j$. For each edge $(u,v) \in E$ and for each $0\leq j < \lmax$, we define, $(u^j , v^{j+1}) \in E'$. For each edge $(u^j,v^{j+1}) \in E'$ we define $C'(u^j,v^{j+1}) = C(u,v)$. 

Note that, by construction, the path-length of all the paths from $s^0_i$ to $e^j_i$ is exactly $j$. For the $i$-th demand we are interested in finding the paths between $s_i$ and $e_i$ with path-length at most $l_i$. Hence, finding paths for the $i$-th demand in $G$ is same as finding paths from $s^0_i$ to $e^1_i, \ldots , e^{l_i}_i$ in the modified network $G'$. For this reason, in the edge-based formulation, we decompose the $i$-th demand $(s_i,e_i,l_i)$ into $l_i$ demands $\{(s^0_i,e^1_i), \ldots , (s^0_i,e^{l_i}_i)\}$ and construct a new demand set called, $\dmod$. It is defined as follows,

\begin{equation}
\label{app_mod_dem}
\dmod := \{D_1, \ldots , D_{k}\},
\end{equation}

where for each $1 \leq i \leq k$, $D_i = \{(s^0_i,e^1_i), \ldots , (s^0_i,e^{l_i}_i)\}$. 
\subsection{Path-based formulation}
\label{app_path_form}

In this section, we rewrite again the path-based formulation based on the new demand set $\dmod$ and the modified network $G'$. Here, for the $i,j$-th demand $(s^0_i,e^j_i) \in \dmod$ we denote $\p_{i,j}$ as the set of all possible paths from $s^0_i$ to $e^j_i$ and for each path $p \in \p_{i,j}$ we define one variable $r_p \in \mathbb{R}^+$. The variable $r_p$ denotes the flow between $s^0_i$ to $e^j_i$ via the path $p$. The aim of the path-based formulation is to maximise the sum $\sum_{i=1}^k \sum_{j=1}^{l_i} \sum_{p \in \p_{i,j}} r_p$.

We give the exact formulation in table \ref{app_path_form}.

\begin{table}[ht]
\begin{mdframed}
\begin{alignat}{3}
\nonumber
&\text{Maximize} \quad \quad \quad \quad \quad \quad \quad \quad\\ \label{app_opt_cr_swap_len1}
&\sum_{i=1}^k \sum_{j=1}^{l_i}\sum_{p \in \p_{i,j}}r_p \quad \quad \quad \quad \quad \quad \\ \nonumber
&\text{Subject to :}\quad \quad \quad \quad \quad \quad \quad \quad \\ \label{app_swp_const1}
&\sum_{i=1}^k \sum_{j=1}^{l_i} \sum_{t=0}^{\lmax -1}\hspace{-0.13in}\sum_{\substack{p \in \p_{i,j} :\\ (u^t,v^{t+1}) \in p}}\hspace{-0.13in}\frac{r_p}{(\pswap)^{|p|-1}} \leq C(u,v), ~ &\forall ~(u,v) \in E,\\\ \nonumber
&\forall i \in \{1, \ldots ,k\}, \forall  j \in \{1, \ldots, l_i\}, \forall ~p\in \p_{i,j}, &\\ \label{app_swp_const2}
&r_p \geq 0.\quad \quad \quad \quad \quad \quad \quad \quad&
\end{alignat}
\vspace{0.1in}
\end{mdframed}
\caption{Path-Based Formulation on the Modified Network}
\label{app_path_form}
\end{table}

\subsection{Edge-based formulation}
\label{app_edge_form}

In this section, for clarity, we rewrite the edge formulation. Here, for the $i,j$-th demand $(s^0_i,e^j_i) \in D_i$ (where $D_i \in \dmod$), we define one function $g_{ij}: E' \rightarrow \R^+$. We give the edge-formulation in table \ref{app_edge_form}.

\begin{table}[ht]
\begin{mdframed}
\begin{alignat}{3}
\nonumber
&\text{Maximize}\\ \label{app_edge_form_plen3}
&\sum_{i=1}^k\sum_{j=1}^{l_i} (\pswap)^{j-1}\sum_{\substack{v^1:(s^0_i,v^1) \in E'}}  g_{ij}(s^0_i,v^1). \\ \nonumber
&\text{Subject to :}\\ \nonumber 
&\text{For all }  1 \leq i \leq k,  1 \leq j \leq  l_{i},  0 \leq t \leq \lmax-1, (u,v) \in E,\\  \label{app_const_edge31}
& g_{ij}(u^t,v^{t+1}) \geq 0. \\ \label{app_const_edge32}
& \sum_{i=1}^k \sum_{j=1}^{l_i} \sum_{t=0}^{\lmax -1}g_{ij}(u^t,v^{t+1}) \leq C(u,v).\\ \nonumber
&\text{For all } 1 \leq i \leq k, 1 \leq j \leq  l_{i},\\ \nonumber
& \text{For all }u',v',w' \in V' : v' \neq s^0_i, v' \neq e^{j}_i , \\ \label{app_const_edge33}
&\sum_{u': (u',v') \in E'} g_{ij}(u',v')= \sum_{w': (v',w') \in E'} g_{ij}(v',w').
\end{alignat}
\vspace{0.1in}
\end{mdframed}
\caption{Edge-based formulation on the modified network}
\label{app_edge_form}
\end{table}

\subsection{From the path-based formulation to the edge-based formulation}
\label{app_pte_prep_swap}

In this section, we show that the solution of the edge-based formulation is at least as good as the solution of the path-based formulation. In order to do so we assume that we have the solution of the path-based formulation proposed in table \ref{app_path_form}. From this solution we construct a solution of the edge-based formulation, proposed in table \ref{app_edge_form}. For the $i,j$-th demand $(s^0,e^j_i) \in \dmod$, for each edge $(u',v') \in E'$, we define a function $\g_{ij}: E \rightarrow \R^+$, as follows,

\begin{equation}
\label{app_eq_gij_prep_swap}
\g_{ij}(u',v') := \sum_{\substack{p \in \p_{i,j}, \\ (u',v') \in p}} \frac{r_p}{(\pswap)^{j-1}}. 
\end{equation}

Here, we show that this $\g_{ij}$ is a valid solution for the edge-based formulation. In order to do so, first we need to show that $\g_{ij}$ corresponds to the objective function of the edge-formulation.  

 \begin{proposition}
\label{obj_pte_prep_swap}
For all $(u',v') \in E'$, if we consider equation \ref{app_eq_gij_prep_swap} as the definition of the function $\g_{ij}:E' \rightarrow \R^+$ then,
\begin{equation}
\sum_{i=1}^k \sum_{j=1}^{l_i}\sum_{p \in \p_{i,j}}r_p = \sum_{i=1}^k\sum_{j=1}^{l_i} (\pswap)^{j-1}\sum_{\substack{v^1:(s^0_i,v^1) \in E'}}  \g_{ij}(s^0_i,v^1).
\end{equation}
\end{proposition}

\begin{proof}
According to the definition of $\g_{ij}$ in equation \ref{app_eq_gij_prep_swap}, we get

\begin{equation}
\g_{ij}(u',v') = \sum_{\substack{p \in \p_{i,j}, \\ (u',v') \in p}} \frac{r_p}{(\pswap)^{j-1}}. 
\end{equation}

This implies, 

\begin{equation}
\sum_{\substack{p \in \p_{i,j}, \\ (u',v') \in p}} r_p= (\pswap)^{j-1}\g_{ij}(u',v').
\end{equation}

By taking the summation over all $1\leq i \leq k$ and $1 \leq j \leq l_i$ at the both side of the above equation we can prove this proposition. 

\end{proof}

In the rest of this section, we show that $\g_{ij}$ satisfies all the constraints from equations \ref{app_const_edge31} to equations \ref{app_const_edge33}. Note that, $\g_{ij}$ satisfies the first constraint of the edge-formulation by construction. In the next proposition, we show that $\g_{ij}$ satisfies the constraint equations: $\forall (u,v) \in E$ and for all $0 \leq t \leq \lmax-1$, $\sum_{i=1}^k \sum_{j=1}^{l_i}\sum_{t=0}^{\lmax -1}g_{ij}(u^t,v^{t+1}) \leq C(u,v)$.

\begin{proposition}
\label{prop_path_prep_cond2}
For all $(u,v) \in E'$, if we consider equation \ref{app_eq_gij_prep_swap} as the definition of the function $\g_{ij}:E' \rightarrow \R^+$ then, for all $1 \leq i \leq k,  1 \leq j \leq  l_{i},  0 \leq t \leq \lmax-1, \forall~(u,v) \in E$
\begin{equation}
 \sum_{i=1}^k \sum_{j=1}^{l_i} \sum_{t=0}^{\lmax -1}\tilde g_{ij}(u^t,v^{t+1}) \leq C(u,v).
\end{equation}
\end{proposition}

\begin{proof}
For any $1 \leq i \leq k$, $1 \leq j \leq l_i$, and $0 \leq t \leq \lmax -1$ and for any edge $(u,v) \in E$, we define the function $\g_{ij}$ in equation \ref{app_eq_gij_prep_swap} as follows. 
\begin{equation}
\g_{ij}(u^t,v^{t+1}) = \sum_{\substack{p \in \p_{i,j}, \\ (u^t,v^{t+1}) \in p}} \frac{r_p}{(\pswap)^{j-1}}. 
\end{equation} 
By taking the sum over all the values of $i$, $j$ and $t$, we get,
\begin{equation}
\label{app_prep_prf_const_2a}
\sum_{i=1}^k\sum_{j=1}^{l_i}\sum_{t=0}^{\lmax -1}\g_{ij}(u^t,v^{t+1}) =\sum_{i=1}^k \sum_{j=1}^{l_i}\sum_{t=0}^{\lmax -1}\sum_{\substack{p \in \p_{i,j}, \\ (u^t,v^{t+1}) \in p}} \frac{r_p}{(\pswap)^{j-1}}.
\end{equation}

From the constraint equation \ref{app_swp_const1} of the path-formulation we get,

\begin{equation}
\sum_{i=1}^k \sum_{j=1}^{l_i}\sum_{t=0}^{\lmax -1}\sum_{\substack{p \in \p_{i,j}, \\ (u^t,v^{t+1}) \in p}} \frac{r_p}{(\pswap)^{|p|-1}} \leq C(u,v).
\end{equation}

Substituting this inequality in equation \ref{app_prep_prf_const_2a} we get,

\begin{equation}
\label{app_prep_prf_const_2a}
\sum_{i=1}^k\sum_{j=1}^{l_i}\sum_{t=0}^{\lmax -1}\g_{ij}(u^t,v^{t+1}) \leq C(u,v).
\end{equation}

This concludes the proof.
\end{proof}

In our next proposition we prove that $\g_{ij}$ satisfies the constraint proposed in equation \ref{app_const_edge33}, which is, for all $1 \leq i \leq k$, $1 \leq j \leq  l_{i}$,
 $u',v',w' \in V' : v' \neq s^0_i$, $v' \neq e^{j}_i$,

\begin{align*}
\sum_{u': (u',v') \in E'} \g_{ij}(u',v')&= \sum_{w': (v',w') \in E'} \g_{ij}(v',w').
\end{align*}

\begin{proposition}
\label{prop_path_prep_cond3}
For all $(u',v') \in E'$, if we consider equation \ref{app_eq_gij_prep_swap} as the definition of the function $\g_{ij}:E' \rightarrow \R^+$ then, for all $1 \leq i \leq k$, $1 \leq j \leq  l_{i}$,
$v' \in V' : v' \neq s^0_i$, $v' \neq e^{j}_i$, 

\begin{align*}
\sum_{u': (u',v') \in E'} \g_{ij}(u',v')&= \sum_{w': (v',w') \in E'} \g_{ij}(v',w').
\end{align*}
\end{proposition}

\begin{proof}
From the definition of $\g_{ij}$ in equation \ref{app_eq_gij_prep_swap} we have,
\begin{equation}
\g_{ij}(u',v') = \sum_{\substack{p \in \p_{i,j}, \\ (u',v') \in p}} \frac{r_p}{(\pswap)^{j-1}}. 
\end{equation}
Using this relation, for any node $v' \in V'$, such that $v' \neq s^0_i, v' \neq e^{j}_i$, we can rewrite the expression $\sum_{u': (u',v') \in E'} \g_{ij}(u',v')$ in a following manner.
\begin{equation}
\label{temp_prep_swap_pte_1}
\sum_{u': (u',v') \in E'} \g_{ij}(u',v') = \sum_{u': (u',v') \in E'}\sum_{\substack{p \in \p_i, \\ (u',v') \in p}} \frac{r_p}{(\pswap)^{j-1}}.
\end{equation}

For each edge $(u',v') \in E'$ (where $v' \in V'\setminus \{s^0_i,e^j_i\}$) can be part of multiple paths $p \in \p_{i,j}$. This implies, 

\begin{align*}
\sum_{\substack{p \in \p_{i,j} :\\ (u',v')\in p}} & \frac{r_p}{(\pswap)^{j-1}} = \sum_{w': (v',w') \in E'} \sum_{\substack{p \in \p_{i,j} :\\ (u',v')\in p \\ (v',w')\in p}} \frac{r_p}{(\pswap)^{j-1}}.
\end{align*}
Substituting the value of $\sum_{\substack{p \in \p_{i,j} :\\ (u',v')\in p}} \frac{r_p}{(\pswap)^{j-1}}$ in equation \ref{temp_prep_swap_pte_1} we get,

\begin{align}
\label{temp2}
\sum_{\substack{u': \\ (u',v') \in E'}} \g_{ij}(u',v')  & =\sum_{\substack{u': \\ (u',v') \in E'}} \sum_{\substack{w': \\ (v',w') \in E'}} \sum_{\substack{p \in \p_{i,j} :\\ (u',v')\in p \\ (v',w')\in p}} \frac{r_p}{(\pswap)^{j-1}}
\end{align}

At the right hand side of the above equation, by interchanging the summation over $u'$ and $w'$ we get,

\begin{align}
\nonumber
\sum_{u': (u',v') \in E'} \g_{ij}(u',v')& = \hspace{-0.15in} \sum_{w': (v',w') \in E'}  \sum_{u': (u',v') \in E'} \sum_{\substack{p \in \p_{i,j} :\\ (u',v')\in p \\ (v',w')\in p}} \frac{r_p}{(\pswap)^{j-1}}
\end{align}
As for an intermediate node $v'$, the total number of paths, incoming to it same as the total number of paths leaving it, so we can rewrite the above expression as, 
\begin{align}
\nonumber
\sum_{u': (u',v') \in E'} \g_{ij}(u',v')& =  \sum_{w': (v',w') \in E'}  \sum_{\substack{p \in \p_{i,j} : \\ (v',w')\in p}} \frac{r_p}{(\pswap)^{j-1}}.
 \end{align}
 
According to the definition of $\g_{ij}$, we have, $\g_{ij}(v',w') = \sum_{\substack{p \in \p_{i,j} : \\ (v',w')\in p}} \frac{r_p}{(\pswap)^{j-1}}$. By substituting this relation on the right hand side of the above expression we get,
\begin{align}
 \nonumber
\sum_{u': (u',v') \in E'} \g_{ij}(u',v') &=  ~~~~\sum_{w': (v',w') \in E'}   \g_{ij}(v',w').
\end{align}

This concludes the proof.

\end{proof}

Proposition \ref{prop_path_prep_cond2} and proposition \ref{prop_path_prep_cond3} certifies that $\g_{ij}$ all the constraints, proposed in the edge-based formulation and proposition \ref{obj_pte_prep_swap} proves that $\g_{ij}$ corresponds to the objective function of the edge-formulation. This implies, $\g_{ij}$ corresponds to a valid solution of the edge-based formulation. In the next section, we show how to construct the path-based formulation from the edge-based formulation.

\subsection{From the edge-based formulation to the path-based formulation}
\label{app_etp_prep_swap}

In this section, we show that the path-based formulation is at least as good as the edge-based formulation. We assume that we have the solution of the edge-based formulation, defined in section \ref{app_edge_form}. From this solution we extract a solution for the path-based formulation. We use algorithm \ref{app_flow_dec3} for extracting the paths and the achievable rates for the path-based formulation. In algorithm \ref{app_flow_dec3}, at step $8$, for the $i,j$-th demand $(s^0_i,e^j_i) \in \dmod$ we compute the entanglement distribution rate $\ra_{p_{j,m}}$ across a path $p_{j,m} \in \p_{i,j}$. In order to be a valid solution of the path-based formulation, proposed in table \ref{app_path_form} we need to show that the extracted rates should satisfy the constraints in equation \ref{app_swp_const1} and equation \ref{app_swp_const2}. We also need to show that, the objective function which is computed from these extracted rates should corresponds to the objective function of the edge-based formulation. In order to do that, first we need to prove some properties of the function $F_{i,j,m}$, used in algorithm \ref{app_flow_dec3}. In the next proposition, we show that for all the edges $(u',v') \in E'$, the value of the function $F_{i,j,m} \geq 0$.


\begin{algorithm}
\begin{mdframed}
\vspace{0.1in}
\hspace*{\algorithmicindent} \textbf{Input: } The solution we obtain from the edge-based formulation, i.e., $\{\{g_{ij}(u',v')\}_{(u',v')\in E'}\}_{1\leq i\leq k, 1 \leq j \leq l_i}$. \\
    \hspace*{\algorithmicindent} \textbf{Output: } Set of paths as well as the rate across each of the paths $\{\p_{i,j}\}_{1\leq i\leq k, 1 \leq j \leq l_i}$.
\begin{algorithmic}[1]
\FOR{($i=1$; $i \leq k$; $i++$)}
  \FOR{($j=1$; $j \leq l_i$; $j++$)}
  \State $m =0$.
  \State $F_{i,j}=g_{ij}$.
  \State $\p_{i,j} = \emptyset$.
 \WHILE{$\sum_{v:(s^0_i,v^1)\in E'} F_{i,j,m}(s^0_i,v^1) >0$}
  \State Find a path $p_{j,m}$ from $s^0_i$ to $e^j_i$ such that,
   \State $\forall(u',v')\in p_{j,m}$, $F_{i,j,m}(u',v') >0$
  \State $\ra_{p_{j,m}} = (\pswap)^{j-1}\min_{(u,v)\in p_{j,m}}\{F_{i,j,m}(u',v')\}$
   \State $\forall (u',v')\in p_{j,m}$, we define $F_{i,j,m+1}(u',v')$ as,
  \State $F_{i,j,m+1}(u',v') := F_{i,j,m}(u',v') - \frac{\ra_{p_{j,m}}}{(\pswap)^{j-1}}$.
  \State $\p_{i,j} = \p_{i,j} \cup (p_{j,m},\ra_{p_{j,m}})$.
    \State $m=m+1$.
  \ENDWHILE
  \ENDFOR
\ENDFOR
\end{algorithmic}
\vspace{0.1in}
\end{mdframed}
 \caption{Path Extraction and Rate Allocation Algorithm.}
 \label{app_flow_dec3}
\end{algorithm}

\begin{proposition}
\label{prop_prep_cond1}
In algorithm \ref{app_flow_dec3} for all $1 \leq i \leq k$, $1 \leq j \leq l_i$, $m \geq 0$, $(u',v') \in E'$,
\begin{equation}
F_{i,j,m}(u',v') \geq 0.
\end{equation}
\end{proposition}

\begin{proof}
In the algorithm \ref{app_flow_dec3}, after each iteration over $m$, we compute $F_{i,j,m}(u',v') = F_{i,j,m-1}(u',v') - \frac{\ra_{p_{j,m-1}}}{(\pswap)^{j-1}}$, where $\ra_{p_{j,m-1}} = (\pswap)^{j-1}\min_{(u',v')\in p_{j,m}}\{F_{i,j,m}(u',v')\}$. This implies, at least for one edge $(u',v') \in E'$, $F_{i,j,m}(u',v') =0$ and for the other edges $(u',v') \in p_{j,m}$,
\begin{align*}
\ra_{p_j} \leq F_{i,j}(u',v')&(\pswap)^{j-1}.
\end{align*}
This implies, 
\begin{align*}
F_{i,j}(u',v') - \frac{\ra_{p_j}}{(\pswap)^{j-1}} &\geq 0.
\end{align*}

This concludes the proof.

\end{proof}

In this next proposition, we show that for each demand $i,j$ the total number of paths $|\p_{i,j}|$ is upper bounded by $|E'|$.

\begin{proposition}
\label{app_prop_prep_cond2}
In algorithm \ref{app_flow_dec3},
\begin{equation}
|\p_{i,j}| \leq |E'|.
\end{equation}
\end{proposition}

\begin{proof}
Due to the flow conservation property of the edge-based formulation, if for some neighbour of $s^0_i$, $F_{i,j,m}(s^0_i,v^1) > 0$ then there exist a path $p_{j,m}$ from $s^0_i$ to $e^j_i$ such that $F_{i,j,m}(u',v') >0$ for all $(u',v') \in p_{j,m}$. Note that, at each step $m$ of the algorithm \ref{app_flow_dec3} there exist at least one edge $(u',v') \in E'$ in the discovered the path $p_{j,m}$, such that $F_{i,j,m+1}(u',v') =0$. As there are in total, $|E'|$ number of edges and the algorithm runs until $\sum_{v^1:(s^0_i,v^1) \in E'}F_{i,j,m+1}(s^0_i,v^1)=0$, so the maximum value of $m$ could not be larger than $|E'|$.
\end{proof}

In the edge-based formulation we have the flow conservation for each $g_{ij}$. In the next proposition we show that the flow conservation also holds for all $F_{i,j,m}$.

\begin{proposition}
\label{prop_prep_cond_int}
In algorithm \ref{app_flow_dec3} for all $1\leq i \leq k$, $1 \leq j \leq l_i$, $m \geq 0$ and $\forall ~ v' \in V' \setminus \{s^0_i,e^j_i\}$ ,
\begin{equation}
  \sum_{\substack{u': \\ (u',v') \in E'}} F_{i,j,m}(u',v')= \sum_{\substack{w': \\ (v',w') \in E'}} F_{i,j,m}(v',w').
\end{equation}
\end{proposition}

In algorithm \ref{app_flow_dec3} for the $i,j$-th demand, we discover a path with each iteration over $m$. We denote the set of discovered paths up to the $m$-th iteration as $\p_{i,j,m}$. After each discovery of the path, we allocate the rate $r_{p_{j,m}}$ across that path using a function $F_{i,j,m}$ and compute the value of the new function $F_{i,j,m+1}$ by subtracting the allocated rate from $F_{i,j,m}$. For $m=0$, we have $F_{i,j,0} = g_{ij}$. This implies, for every iteration $m$, we can rewrite $g_{ij}$ as a function of $F_{i,j,m}$ and the sum of the allocated rates so far. In the edge formulation, the function $g_{ij}$ satisfies the flow conservation property. Here we use this relation and substitute $g_{ij}$ with the function of $F_{i,j,m}$, then by doing some simple algebraic manipulation we could show that $F_{i,j,m}$ also satisfies the flow conservation property. 

\begin{proof}
In algorithm \ref{app_flow_dec3} suppose for any $m \geq 0$, the set of discovered paths are $\p_{i,j,m}$. This implies, for any edge $(u',v') \in E'$,

\begin{align}
\nonumber
F_{i,j,m}(u',v') &= g_{ij}(u',v') - \sum_{p_{j,m} \in P_{i,j,m}} \frac{\ra_{p_{j,m}}}{(\pswap)^{j-1}}
\end{align}
By exchanging the position of $g_{ij}(u',v')$ and $F_{i,j,m}(u',v')$ in the above equation we get,
\begin{align}
\label{g__prep_temp1}
g_{ij}(u',v') &= F_{i,j,m}(u',v') + \sum_{\substack{p_{j,m} \in P_{i,j,m}: \\ (u',v') \in p_{j,m}}} \frac{\ra_{p_{j,m}}}{(\pswap)^{j-1}}.
\end{align}

From the flow conservation property (equation \ref{app_const_edge33}) of the edge formulation we have,

\begin{align*}
 \sum_{\substack{u': \\ (u',v') \in E'}} g_{ij}(u',v')=  \sum_{\substack{w': \\ (v',w') \in E'}} g_{ij}(v',w').
\end{align*}

Substituting the value of $g_{ij}(u',v')$ from the equation \ref{g__prep_temp1} we get,

\begin{align}
\nonumber
\sum_{u: (u',v') \in E'} & \left(F_{i,j,m}(u',v') + \sum_{\substack{p_{j,m} \in P_{i,j,m}: \\ (u',v') \in p_{j,m}}} \frac{\ra_{p_{j,m}}}{(\pswap)^{j-1}} \right)\\ \label{prep_temp_flow2}
=   \sum_{w': (v',w') \in E'} &\left(F_{i,j,m}(v',w') + \sum_{\substack{p_{j,m} \in P_{i,j,m}: \\ (v',w') \in p_{j,m}}} \frac{\ra_{p_{j,m}}}{(\pswap)^{j-1}} \right).
\end{align}

As for an intermediate node $v'$, the number of the incoming paths to it is same as the number of outgoing paths from it. This implies, 

\begin{align*}
\sum_{\substack{p_{j,m} \in P_{i,j,m}: \\ (u',v') \in p_{j,m}}} \frac{\ra_{p_{j,m}}}{(\pswap)^{j-1}}  = \sum_{\substack{p_{j,m} \in P_{i,j,m}: \\ (v',w') \in p_{j,m}}}  \frac{\ra_{p_{j,m}}}{(\pswap)^{j-1}} .
\end{align*}
By substituting this relation in equation \ref{prep_temp_flow2} we get,

\begin{align*}
\sum_{u': (u',v') \in E'} & F_{i,j,m}(u',v') = \sum_{w': (v',w') \in E'} F_{i,j,m}(v',w').
\end{align*}

This concludes the proof.

\end{proof}

In the next proposition, we show that, the rates we compute in algorithm \ref{app_flow_dec3} satisfies the condition \ref{app_swp_const1} of the path-based formulation.

\begin{proposition}
\label{prop_prep_cond3}
For all $(u,v) \in E$,
\begin{align}
\sum_{i=1}^k \sum_{j=1}^{l_i}\sum_{t=0}^{\lmax -1}\sum_{\substack{p_{j,m} \in \p_{i,j} :\\ (u^t,v^{t+1}) \in p_{j,m}}}&\frac{\ra_{p_{j,m}}}{(\pswap)^{j-1}}\leq C(u,v),
\end{align}
where $\ra_{p_{j,m}}$ is defined in algorithm \ref{app_flow_dec3} (see step $8$).
\end{proposition}

In algorithm \ref{app_flow_dec3} for the $i,j$-th demand, we discover a path with each iteration over $m$. We denote the set of discovered paths up to the $m$-th iteration as $\p_{i,j,m}$. After each discovery of the path, we allocate the rate $\ra_{p_{j,m}}$ across that path using a function $F_{i,j,m}$ and compute the value of the new function $F_{i,j,m+1}$ by subtracting the allocated rate from $F_{i,j,m}$. For $m=0$, we have $F_{i,j,0} = g_{ij}$. This implies, for every iteration $m$, we can rewrite $g_{ij}$ as a function of $F_{i,j,m}$ and the sum of the allocated rates so far. In the edge formulation, for every edge $(u,v) \in E$ and for every $0\leq t \leq \lmax -1$ we have that $\sum_{i=1}^k\sum_{j=1}^{l_i} \sum_{t=0}^{\lmax -1}g_{ij}(u^t,v^{t+1}) \leq C(u,v)$. Here, we use this relation and substitute $g_{ij}$ with the function of $F_{i,j,m}$, then by doing some simple algebraic manipulation we could show that the sum of the extracted rate is also upper bounded by the capacity of that edge. 

\begin{proof}

In algorithm \ref{app_flow_dec3}, suppose for any $m \geq 0$ and for any $1 \leq i \leq k$, $1 \leq j \leq l_i$, and $0\leq t \leq \lmax-1$, the set of discovered paths is $\p_{i,j,m}$. This implies, for any edge $(u,v) \in E$, and for any $0 \leq t \leq \lmax -1$

\begin{align}
\nonumber
F_{i,j,m}(u^t,v^{t+1}) &= g_{ij}(u^t,v^{t+1}) - \sum_{\substack{p_{j,m} \in P_{i,j,m} \\ (u^t,v^{t+1}) \in p_{j,m}}} \frac{\ra_{p_{j,m}}}{(\pswap)^{j-1}}.
\end{align}

From proposition \ref{prop_prep_cond1} we have that for all $0\leq m \leq |E'|$ and for all the edges $(u^t,v^{t+1}) \in E'$, $F_{i,j,m}(u^t,v^{t+1}) \geq 0$. This implies, 

\begin{align}
\nonumber
g_{ij}(u^t,v^{t+1}) - \sum_{\substack{p_{j} \in P_{i,j}: \\ (u^t,v^{t+1}) \in p_{j}}} \frac{\ra_{p_{j}}}{(\pswap)^{j-1}} \geq 0\\
\label{g__prep_temp2}
g_{ij}(u^t,v^{t+1}) \geq \sum_{\substack{p_{j} \in P_{i,j}: \\ (u^t,v^{t+1}) \in p_{j}}} \frac{\ra_{p_{j}}}{(\pswap)^{j-1}}. 
\end{align}

From the edge-based formulation we have, 

$$\sum_{i=1}^k\sum_{j=1}^{l_i} \sum_{t=0}^{\lmax -1} g_{ij}(u^t,v^{t+1}) \leq C(u,v),$$

for all the edges $(u,v) \in E$ and for all $0 \leq t \leq \lmax -1$. Substituting this relation in equation \ref{g__prep_temp2} we get, 

\begin{align}
\label{f_prep_temp1}
\sum_{i=1}^k\sum_{j=1}^{l_i} \sum_{t=0}^{\lmax -1} \sum_{\substack{p_{j} \in P_{i,j}: \\ (u^t,v^{t+1}) \in p_{j}}} \frac{\ra_{p_{j}}}{(\pswap)^{j-1}} & \leq C(u,v).
\end{align}

This concludes the proof.

\end{proof}

We finish this section by showing the equivalence of the objective functions for both of the formulations.

\begin{proposition}[Equivalence of the objective functions]
\label{prop_cond3}
In algorithm \ref{app_flow_dec3},
\begin{equation}
\sum_{i=1}^k\sum_{j=1}^{l_i} (\pswap)^{j-1}\sum_{\substack{v^1:(s^0_i,v^1) \in E'}}  g_{ij}(s^0_i,v^1) = \sum_{i=1}^k\sum_{j=1}^{l_i}\sum_{p \in \p_{i,j}}\ra_{p}.
\end{equation}
\end{proposition}

Due to the flow conservation property of the function $F_{i,j,m}$, the algorithm \ref{app_flow_dec3} runs until, $F_{i,j,m+1}(s^0_i,v^1) =0$, for all the neighbours of the source node $s^0_i$. In the previous propositions, we establish a relation between $g_{ij}$ and $F_{i,j,m}$ and the set of discovered paths, i.e., $F_{i,j,m}(s^0_i,v^{1}) = g_{ij}(s^0_i,v^{1}) - \sum_{\substack{p_{j,m} \in P_{i,j,m} \\ (s^0_i,v^{1}) \in p_{j,m}}} \frac{\ra_{p_{j,m}}}{(\pswap)^{j-1}}$. If all the paths are discovered, then the value of $F_{i,j,m+1}(s^0_i,v^1)$ becomes zero and $g_{ij}$ will only be a function of all the discovered paths. Here, we use this relation to prove the equivalence of the objective functions.

\begin{proof}
From proposition \ref{prop_prep_cond_int} we have that the functions $F_{i,j,m}$ follow the flow conservation. From the proposition \ref{prop_prep_cond1} we have that $F_{i,j,m}(u',v') \geq 0$ for all $0 \leq m \leq |E'|$. This implies, for a fixed $i,j$, for any value of $m$, we can always find one path $p_{j,m}$ with non-zero $\ra_{p_{j,m}}$ until, $\sum_{v^1:(s^0_i,v^1) \in E'} F_{i,j,m}(s^0_i,v^1) =0$. According to the algorithm \ref{app_flow_dec3} after each iteration if an edge $(s^0_i,v^1) \in p_{j,m}$,

\begin{equation}
\label{obj_prep_1}
F_{i,j,m+1}(s^0_i,v^1)  = F_{i,j,m}(s^0_i,v^1) - \frac{\ra_{p_{j,m}}}{(\pswap)^{j-1}}.
\end{equation}

If for some $i,j,m$, $F_{i,j,m+1}(s^0_i,v^1) =0$ then from the equation \ref{obj_prep_1} we have,

\begin{align*}
F_{i,j,m}(s^0_i,v^1) &=\frac{\ra_{p_{j,m}}}{(\pswap)^{j-1}}.
\end{align*}

Using the recurrence relation of equation \ref{obj_prep_1} we can rewrite the above expression as,

\begin{align*}
F_{i,j,m-1}(s^0_i,v^1) -\frac{\ra_{p_{j,m-1}}}{(\pswap)^{j-1}} &= \frac{\ra_{p_{j,m}}}{(\pswap)^{j-1}}.
\end{align*}
If we continue like this until $m=0$, then we get
\begin{align*}
F_{i,j,0}(s^0_i,v^1) &= \sum_{m'=0}^m \frac{\ra_{p_{j,m'}}}{(\pswap)^{j-1}}.
\end{align*}
Note that, $F_{i,j,0} = g_{ij}$. This implies, 
\begin{align}
\nonumber
g_{ij}(s^0_i,v^1) &= \sum_{m'=0}^m  \frac{\ra_{p_{j,m'}}}{(\pswap)^{j-1}}\\ \label{obj_2}
\sum_{m'=0}^m  \ra_{p_{j,m'}} & = (\pswap)^{j-1} g_{ij}(s^0_i,v^1),
\end{align}
This is true for all $m$. This implies,

\begin{equation}
(\pswap)^{j-1}g_{ij}(s^0_i,v^1) = \sum_{\substack{p \in P_{i,j} \\ (s^0_i,v^1) \in p}}  \ra_{p}.
\end{equation}
By taking sum on both sides on $i$ and $j$ we get,

\begin{equation}
\sum_{i=1}^k\sum_{j=1}^{l_i}(\pswap)^{j-1}\sum_{\substack{v:(s^0_i,v^1) \in E'}}  g_{ij}(s^0_i,v^1)= \sum_{i=1}^k\sum_{j=1}^{l_i}\sum_{p \in \p_{i,j}}\ra_{p}.
\end{equation}
This concludes the proof.

\end{proof}

\section{Prepare and swap protocol}
\label{app_prep_swap}

In this appendix we prove the EPR-pair generation rate across a repeater chain using prepare and swap protocol.

\begin{lemma}
\label{app_egr_prep_swap}

In a repeater chain network with $n+1$ repeaters $\{u_0, u_1, \ldots, u_{n}\}$, if the probability of generating an elementary pair per attempt is one and the probability of a successful BSM is $(\pswap)$  and the capacity of an elementary link $(u_i,u_{i+1})$ (for $0 \leq i \leq n-1$) is denoted by $C_i$ and if the repeaters follow the prepare and swap protocol for generating EPR-pairs, then the expected end-to-end entanglement generation rate $r_{u_0,u_{n}}$ is,
\begin{align}
\label{supp_eg_rate_prep_mes}
r_{u_0,u_{n}} & = (\pswap)^{n-1}\min\{C_0, \ldots ,C_{n-1}\}.
\end{align}
\end{lemma}

\begin{proof}
In the entanglement generation protocol, first, the repeaters start generating the elementary pairs in parallel. As, the elementary pair generation is a deterministic event, so each of the node $u_i$ can generate $C_i$ EPR-pairs with its neighbour $u_{i+1}$ ($0 \leq i \leq n-1$) per second. After generating the elementary pairs, the intermediate nodes perform the swap operations independently of each other. This implies, if all the swap operations are deterministic then the end-to-end entanglement generation is $\min\{C_1, \ldots, C_n\}$.
However, each of the swap operations succeed with probability $(\pswap)$. The end-to-end entanglement generation probability is equal to the probability that all elementary links are successfully swapped which is $(\pswap)^n$ and there are $\min\{C_1, \ldots, C_n\}$ such elementary links. This implies, the expected end-to-end entanglement generation rate is $r_{u_0,u_{n}} = (\pswap)^{n-1}\min\{C_0, \ldots ,C_{n-1}\}$. This concludes the proof.
\end{proof}

\end{document}